\documentclass[reqno,11pt]{amsart}
\usepackage[utf8]{inputenc}
\usepackage{graphicx}
\usepackage{slashed}
\usepackage{amscd}
\usepackage{amssymb}
\usepackage{esint}
\usepackage{enumitem}
\usepackage{braket}	
\usepackage{comment}
\usepackage{amsmath}
\usepackage{dsfont}
\usepackage[mathscr]{eucal}
\usepackage[linktocpage]{hyperref}
\usepackage{pst-all}

\numberwithin{equation}{section}
\allowdisplaybreaks[1]
\definecolor{labelkey}{gray}{.65}

\title[The Relative Fermionic Entropy in Rindler Spacetime]{The Relative Fermionic Entropy in \\
Two-Dimensional Rindler Spacetime}

\author[F.\ Finster]{Felix Finster}
\address{Fakult\"at f\"ur Mathematik  Universit\"at Regensburg  D-93040 Regensburg  Germany}
\email{finster@ur.de}

\author[A.\ Much]{Albert Much \\ \\ March 2026} 
\address{Institut f\"ur Theoretische Physik\\ Universit\"at Leipzig\\ D-04103 Leipzig \\ Germany}
\email{much@itp.uni-leipzig.de}

\newtheorem{Def}{Definition}[section]
\newtheorem{Thm}[Def]{Theorem}
\newtheorem{Prp}[Def]{Proposition}
\newtheorem{Lemma}[Def]{Lemma}
\newtheorem{Remark}[Def]{Remark}

\newcommand{\Thanks}{\vspace*{.5em} \noindent \thanks}
\newcommand{\beq}{\begin{equation}}
\newcommand{\eeq}{\end{equation}}
\newcommand{\Proof}{\begin{proof}}
	\newcommand{\QED}{\end{proof} \noindent}
\newcommand{\QEDrem}{\ \hfill $\Diamond$}

\newcommand{\la}{\langle}
\newcommand{\ra}{\rangle}

\newcommand{\Sl}{\mbox{$\prec \!\!$ \nolinebreak}}
\newcommand{\Sr}{\mbox{\nolinebreak $\succ$}}

\newcommand{\C}{\mathbb{C}}
\newcommand{\R}{\mathbb{R}}
\newcommand{\1}{\mathds{1}}

\DeclareMathOperator{\tr}{tr}

\newcommand{\Cisc}{C^\infty_{\text{\rm{sc}}}}

\newcommand{\Dir}{{\mathcal{D}}}
\renewcommand{\H}{\mathscr{H}}

\newcommand{\Lin}{\text{\rm{L}}}
\newcommand{\F}{{\mathscr{F}}}

\newcommand{\scrM}{\mycal M}
\newcommand{\scrR}{\mycal R}

\newcommand{\scrU}{{\mathscr{U}}}

\newcommand{\bitem}{\begin{itemize}[leftmargin=2.5em]}
\newcommand{\eitem}{\end{itemize}}

\newcommand{\Fock}{{\mathcal{F}}}
\newcommand{\f}{{\mathfrak{f}}}

\newcommand{\A}{\mycal A}

\DeclareFontFamily{OT1}{rsfso}{}
\DeclareFontShape{OT1}{rsfso}{m}{n}{ <-7> rsfso5 <7-10> rsfso7 <10-> rsfso10}{}
\DeclareMathAlphabet{\mycal}{OT1}{rsfso}{m}{n}

\begin{document}
\maketitle
\begin{abstract}
The fermionic relative entropy in two-dimensional Rindler spacetime is studied using both modular theory and the reduced one-particle density operators. The methods and results are compared. A formula for the relative entropy for general Gaussian states is derived. As an application, the relative entropy is computed for a class of
non-unitary excitations.
\end{abstract}
\tableofcontents

\section{Introduction} \label{secintro}
The purpose of this article is to study the fermionic relative entropy from two different perspectives: Using {\em{modular theory}} and working with the {\em{reduced one-particle density operators}}.  In order to compare the two perspectives, we  study here relativistic quantum systems which are so simple that they can be analyzed explicitly, namely {\em{quasi-free fermionic states}} in {{two-dimensional Rindler spacetime}}.
Comparing the methods and results gives some insight into how these methods are interrelated.
Moreover, we work out that the scope and applicability of the approaches is rather different,
meaning that, depending on the precise setup and the question of interest, 
 the methods complement and nourish one another.

We introduce the problem in the setting of algebraic quantum field theory.
Restricting attention to fermionic theories, we let~$\A$ be an algebra of fields satisfying the
canonical anti-commutation relations (CAR algebra).
Thus, given a complex Hilbert space~$(\H, \la .|. \ra)$
(also referred to for clarity as the {\em{one-particle Hilbert space}}),
to every vector~$\phi \in \H$ we associate the operator~$\Psi(\overline{\phi})$ as well as its adjoint~$\Psi^\dagger(\phi)$ satisfying the relations
\beq \label{CAR}
	\big\{ \Psi(\overline{\psi}), \Psi^\dagger(\phi) \big\} = \langle \psi | \phi \rangle \qquad\text{and}\qquad
	\big\{ \Psi(\overline{\psi}), \Psi(\overline{\phi}) \big\} = 0 = \big\{ \Psi^\dagger(\psi), \Psi^\dagger(\phi) \big\} \:.
\eeq
Then~$\A$ is defined as the unital $*$-algebra generated by all these operators
(unital means that~$\A$ contains the identity operator denoted by~$\1$).
A {\em{state}}~$\omega : \A \rightarrow \C$ is a linear functional on~$\A$ which is normalized and positive
in the sense that
\[ 
\omega(\1)=1 \qquad \text{and} \qquad \omega \big( A^* A \big) \geq 0 \quad \forall\; A \in \A\:. \]
In this paper, we restrict our attention to \emph{quasi-free states}, also referred to
as {\em{Gaussian states}}, which are fully determined by their two-point distributions.
More precisely, for a quasi-free state all odd $n$-point distributions vanish, whereas all even $n$-point
distributions can be computed using Wick's theorem.
A relevant subclass of quasi-free states are states which are {\em{particle-number preserving}}, meaning that
all two-point expectations involving two creation or two annihilation operators vanish,
\beq \label{pnp}
\omega\big( \Psi^\dagger(\phi)\, \Psi^\dagger(\psi) \big) = 0 =  \omega\big( \Psi(\overline{\phi})\, \Psi(\overline{\psi}) \big) \:.
\eeq
In the literature, this property is sometimes referred to as a {\em{gauge-invariant state}}
(see~\cite[Proposition~17.32]{derzinski-gerard-quantum}).
A quasi-free and particle-number preserving state is fully determined by its two-point distribution
\beq \label{omega2def}
\omega_2(\overline{\psi}, \phi) := \omega\big( \Psi^\dagger(\phi)\, \Psi(\overline{\psi}) \big) \:.
\eeq
In this case, the {\em{reduced one-particle density operator}}~$D$ is defined 
as the linear operator on the Hilbert space~$(\H, \la .|. \ra)$ with the property
\beq \label{Ddef}
\omega_2(\overline{\psi}, \phi) = \langle \psi | D \phi\rangle \qquad \text{for all~$\psi, \phi \in \H$} \:.
\eeq
It is by construction a symmetric operator which is positive and bounded by one,
\[ 0 \leq D \leq 1 \:. \]
The point of interest is that the {\em{von Neumann entropy}}~$S(\omega)$ of the quantum state~$\omega$
can be expressed in terms of the one-particle density operator by
\[ 
S(\omega) = -\tr \Big( D\, \log(D) + (1-D)\, \log(1-D) \Big) \:. \]
This formula appears commonly in the literature
(see for example~\cite[eq.~(6.3)]{ohya-petz},~\cite{klich, casini-huerta, longo-xu}
and~\cite[eq.~(34)]{helling-leschke-spitzer}).
This formula can be extended to {\em{entanglement }} and {\em{relative entropies}}.
For details in an expository style we refer to the recent survey paper~\cite{fermientropy},
where detailed derivations of these formulas are given in the appendix.
Based on methods first developed in~\cite{widom1},
this setting has been studied extensively for a free Fermi gas formed of non-relativistic spinless
particles~\cite{helling-leschke-spitzer, leschke-sobolev-spitzer, LSS_2022}.
The main interest of these studies lies in the derivation of {\em{area laws}}, which quantify
how the entanglement entropy scales as a function of the size of the spatial region forming the
subsystem. More recently, these methods and results were adapted to the
relativistic setting of the Dirac equation. In~\cite{arealaw} the entanglement entropy for the free Dirac field in
a bounded spatial region of Minkowski space is studied.
An area law is proven in two limiting cases: that the volume tends to infinity
and that the regularization is removed. 
Moreover, in~\cite{diamondentropy} a causal diamond in two-dimensional Minkowski space is considered.
Finally, in~\cite{bhentropy}
the geometry of a Schwarzschild black hole is studied. The entanglement entropy of the event horizon is
computed to be a prefactor times the number of occupied angular momentum modes.
Independently, the entanglement entropy for systems of Dirac spinors has been
studied in~\cite{bollmann-mueller, bollmann-mueller2}.

{\em{Modular theory}} provides a rather different abstract setting in which entropies can be studied
(for general and detailed expositions see~\cite{borchers} or~\cite[Section V]{haagloc}).
The starting point is a {\em{Hilbert space}}~$(\Fock, \la .|. \ra)$ as well as a {\em{von Neumann algebra}}~$\mathcal{M}$ acting on this Hilbert space (for a better distinction from the one-particle Hilbert space,
we prefer the notation~$\Fock$, clarifying that this Hilbert space corresponds to the Fock space
in the many-particle picture).
We denote by~$\mathcal{M}'$ the commutant of~$\mathcal{M}$. Furthermore, we 
assume that there is a vector~$\Omega \in \Fock$ 
which is {\em{cyclic}} and {\em{separating}}, meaning that both  ~$\mathcal{M}\Omega$ and~$\mathcal{M}'\Omega$ are dense
in~$\Fock$. Under these assumptions, there exists a unique anti-linear operator~$S: \Fock \supseteq\text{dom}(S)\rightarrow \Fock$ with the property
\[ SA\,\Omega = A^*\Omega \qquad \qquad \forall \;A\in\mathcal{M} \:, \]
referred to as the {\em{Tomita operator}}.
Since~$S^2=1$, the Tomita operator is invertible and thus has a unique polar decomposition 
\[ S=J\Delta^{\frac{1}{2}} \:, \]
where the linear self-adjoint operator~$\Delta:\Fock \supseteq\text{dom}(\Delta)\rightarrow \Fock$ is referred to as the  {\em{modular operator}},
whereas the anti-unitary operator~$J:\Fock \rightarrow \Fock$ is called
the {\em{modular conjugation}} associated to the pair~$(\mathcal{M},\Omega)$.

The modular operator is self-adjoint, positive and invertible. Thus, it defines
a group of automorphisms, the {\em{modular group}}~$(\sigma_t^{\Omega})_{t \in \R}$ acting on the von Neumann algebra~$\mathcal{M}$ and of its commutant,
\[ \sigma_t^{\Omega} : \mathcal{M} \rightarrow \mathcal{M}\:,\quad
\sigma_t^{\Omega} : \mathcal{M}' \rightarrow \mathcal{M}' \qquad \text{with} \qquad
\sigma_t^{\Omega}(A) := \Delta^{it}A\Delta^{-it} \:. \]
This group is used to formulate the KMS condition characterizing thermal states; see~\cite{haagloc}.

The \textit{relative modular group} generalizes the standard modular flow when considering \textit{two} states rather than a single one. Given two cyclic and separating states $\omega$ and $\tilde{\omega}$ with their respective vector representatives~$\Omega, \tilde{\Omega} \in \Fock$, one defines the  {\em{relative modular operator}} \( \Delta_{\tilde{\Omega}|\Omega} : \Fock \rightarrow \Fock\) by
\[
\Delta_{\tilde{\Omega}|\Omega} = S_{\tilde{\Omega}|\Omega}^* S_{\tilde{\Omega}|\Omega} \:,
\]
where \( S_{\tilde{\Omega}|\Omega} : \Fock \rightarrow \Fock\), the {\em{relative Tomita operator}}, is defined  by its action
\[
S_{\tilde{\Omega}|\Omega} A \Omega := A^* \tilde{\Omega} \quad \forall A \in \mathcal{M} \:.
\]
The {\em{relative modular group}}~$(\sigma_t^{\tilde{\Omega}|\Omega})_{t \in \R}$
is defined by the corresponding flow
\[
\sigma_t^{\tilde{\Omega}|\Omega}(A) = \Delta_{\tilde{\Omega}|\Omega}^{it}\, A\,
\Delta_{\tilde{\Omega}|\Omega}^{-it} \:.
\]
This flow makes it possible to describe the dynamics of observables {\em{relative}} to a reference
state. It also gives rise to a corresponding relative entropy defined as follows.
We denote the states corresponding to~$\Omega$ and~$\tilde{\Omega}$ by~$\omega$ and~$\tilde{\omega}$
(thus in bra/ket notation, $\omega = | \Omega \ra \la \Omega|$ and similarly for~$\tilde{\omega}$).
Assume that the two states are unitarily related, i.e.\ $\tilde{\Omega} = U \Omega$,
where~$U$ is a unitary operator on~$\Fock$. Moreover, we need to assume that
the domain of the operator~$\Delta^{\frac{1}{2}}$ is invariant under the action of the
unitary operator, i.e.~$U\text{dom}(\Delta^{\frac{1}{2}})\subset \text{dom} \Delta^{\frac{1}{2}}$.
Under these assumptions, the {\em{Araki-Uhlmann relative entropy}}~\cite{Araki, Uhlmann} takes the form 
\begin{align*}
    S(\tilde{\omega}\Vert\omega)&=
\langle   \Omega |  \log(\Delta_{\tilde{\Omega}|\Omega})\, \Omega\rangle 
\\&=-\langle U \Omega |  \log(\Delta)\,U\Omega\rangle 
\:. 
\end{align*}  
This quantity was calculated in various settings in quantum field theory~\cite{Hollands_2019, longo, longo-xu, hollands-sanders, hollands-ishibashi, dangelo,  kurpicz, froeb, froeb-much, dangelo-froeb, cadamuro-froeb, La_Piana_2025},
also in relation with the Bekenstein bound~\cite{casini}.
General introductions and reviews are provided by~\cite{hollands-entropy, witten-entangle}.

In order to shed some light on the connections between these approaches, we consider
a simple quantum system to which both modular theory and the methods for the reduced one-particle density
operator apply: free {\em{Majorana fermions in two-dimensional Rindler spacetime}}
(massive or massless). In this setting, one chooses the above cyclic vector~$\Omega$
as the Fock vacuum of two-dimensional Minkowski space.
As~$\mathcal{M}$ one chooses the algebra of fermionic field operators in the Rindler wedge
(the basics on two-dimensional Minkowski and Rindler spacetime will be provided in Section~\ref{secprelim}).
As a consequence, the quantum state~$\omega$ is a thermal state.
Choosing~$\tilde{\omega}$ as an excitation of a specific form,
the relative entropy was computed in~\cite{galanda2023relative}.
One purpose of the present paper is to re-derive these results via the reduced one-particle density operator.
This makes it possible to compare the different methods and results in detail. We also discuss the scope of the
different methods.

In general terms, the connection between modular theory and the one-particle density operator~$D$
in~\eqref{Ddef} is made by the relation
\beq \label{DelD}
D = (1+\Delta)^{-1} \big|_{\Fock^1}\:,
\eeq
where~$\Fock^1 \simeq \H$ denotes the one-particle subspace;
see~\cite{araki1970quasifree}, \cite[Section 3.3]{longo-xu} and~\cite{casini-huerta2}
(in~\cite{longo-xu}, the operator~$D$ is referred to as the covariance operator~$C$).
Moreover, the modular operator can be written as~$\Delta=\exp(-H)$, where~$H$ is the
so-called {\em{modular Hamiltonian}}.
In our example of the Rindler wedge, the modular Hamiltonian is given by~$2\pi K$, where~$K$ is the 
generator of Lorentz boosts. This result is known as the Bisognano-Wichmann
theorem~\cite{bisognano-wichmann1, bisognano-wichmann2}.

With the relation~\eqref{DelD} in mind, our first task is to compute the spectrum and eigenfunctions
of the reduced one-particle density operator~$D$. This can be accomplished
by adapting Fourier methods as developed in~\cite{rindler} for the computation of the fermionic signature
operator (Section~\ref{secrindlerred}).
The next step is to describe the excitation in terms of one-particle density operators.
Here the problem arises that the excited state is no longer particle-number preserving
(meaning that the relation~\eqref{pnp} is violated). This means that the formulas
derived in~\cite[Appendix~A]{fermientropy} can no longer be used.

This problem can be resolved in two alternative ways.
The first method is to compute the relative entropy directly on the Fock space
(Section~\ref{secrel1} and Appendix~\ref{appA}).
The second, alternative method is to extend the formalism of the
reduced one-particle density operator to general (i.e.\ non particle-number preserving)
Gaussian states (Section~\ref{secaltgauss} and Appendix~\ref{appB}).
Clearly, all these methods give the same results, which also agree with the corresponding
computations using modular theory as worked out in Section~\ref{secmodular}.
The comparison of these different computation is quite instructive and gives a better understanding
for how modular theory is related to relative entropy.
As a fundamental prerequisite, the application of modular theory relies on the existence of a cyclic and separating vector~$\Omega$ for the local von Neumann algebra. If the states under consideration are related by a unitary operator belonging to this algebra, the relative modular operator is unitarily equivalent to the modular operator of the reference state. In such scenarios, the relative entropy often admits an explicit analytic evaluation, for instance via the Bisognano-Wichmann theorem. However, unitary implementability is merely a sufficient condition for computational tractability rather than a limitation of the modular framework, which has been successfully applied to compute relative entropies for much broader, non-unitary classes of states \cite{LongoMorsella2023}, \cite{CiolliLongoRuzzi2020}. To explore regimes where explicit modular computations become challenging, Section~\ref{secbeyond} introduces specific examples of non-unitary excitations. For these cases, we demonstrate that the relative entropy can instead be explicitly and rigorously computed utilizing the formalism of reduced one-particle density operators.

The paper is organized as follows. Section~\ref{secprelim} provides the necessary background
on Majorana spinors in two-dimensional Minkowski and Rindler spacetime.
In Section~\ref{seconeparticledensity} the relative entropy is computed via the reduced
one-particle density operators. Section~\ref{secmodular} gives corresponding computations
using modular theory, and the section is concluded by a comparison of the methods and results.
In Section~\ref{secbeyond} we consider non-unitary excitations of general quasi-free particle-preserving states.
Section~\ref{seccompare} is devoted to a brief comparison and an outlook.
The appendices provide detailed computations and supplementary material.

\section{Preliminaries and general setup} \label{secprelim}
\subsection{Majorana spinors in two-dimensional Minkowski space} \label{secmajorana}
We let~$\scrM=\R^{1,1}$ be two-dimensional Minkowski space. Denoting the coordinates with~$(t, x)$,
the line element is given by
\beq \label{mink}
ds^2 = g_{ij}\: dx^i dx^j = -dt^2 + dx^2 \:.
\eeq
We let~$S\scrM = \scrM \times \C^2$ be the trivial spinor bundle.
We work in the so-called Majorana representation of the Dirac matrices
\beq \label{gammarep}
\gamma^0 = \begin{pmatrix} 0 & i \\ -i & 0 \end{pmatrix} \:,\qquad
\gamma^1 = \begin{pmatrix} 0 & i \\ i & 0 \end{pmatrix} \:.
\eeq
The Dirac matrices are symmetric with respect to the {\em{spin inner product}} defined by
\[ 
\Sl \psi | \phi \Sr = \la \psi, \begin{pmatrix} 0 & i \\ -i & 0 \end{pmatrix} \phi \ra_{\C^2} \]
(where~$\la .,. \ra_{\C^2}$ is the canonical inner product on~$\C^2$).
The spin inner product is an indefinite inner product of signature~$(1,1)$.
Introducing the {\em{Dirac operator}}
\beq \label{Dirop}
\Dir := i \gamma^j \partial_j \:,
\eeq
the {\em{Dirac equation}} reads
\beq \label{Direq}
(\Dir-m) \psi = 0 \:,
\eeq
where~$m \geq 0$ is the rest mass (we always work in natural units~$\hbar=c=1$).
Taking smooth and compactly supported initial data on a Cauchy surface~$t=\text{const}$
and solving the Cauchy problem, one obtains a Dirac solution in the class~$\Cisc(\scrM, S\scrM)$
of smooth wave functions with spatially compact support. On solutions~$\psi, \phi$
in this class, one defines the scalar product
\beq \label{print}
(\psi | \phi)_\scrM := \int_{-\infty}^\infty \Sl \psi | \gamma^0 \phi\Sr|_{(t,x)}\: dx \:.
\eeq
Taking the completion, we obtain the Hilbert space denoted by~$(\H_\scrM, (.|.)_\scrM)$.

In order to simplify the setting further, from now on we consider {\em{Majorana spinors}}.
Following the procedure in~\cite[Appendix~A]{octonions}, we recover them as a special
class of Dirac solutions. The pseudo-scalar matrix becomes
\[ \gamma^5 := i \gamma^0 \gamma^1 = \begin{pmatrix} -i & 0 \\ 0 & i \end{pmatrix} \:. \]
We now consider the Dirac operator with a mass~$m$ and pseudo-scalar mass~$n$,
\[ 
\big( i \gamma^j \partial_j + i \gamma^5 n - m \big) \psi = 0 \:. \]
Choosing both masses to be real, all the matrix entries on the left are real, as one sees by
writing the Dirac equation as
\[ \begin{pmatrix} n-m & -\partial_t-\partial_x \\ \partial_t-\partial_x & -n-m \end{pmatrix} \psi = 0 \:. \]
Therefore, the equation admits real-valued solutions, i.e.\
\beq \label{majorana}
\psi(x) \in \R^2 \:.
\eeq
Restricting attention to solutions of this form, the Dirac equation reduces to the Majorana equation.

Throughout this paper, we will restrict attention to the case~$n=0$. In this case, the Majorana
equation simplifies to
\[ 
\begin{pmatrix} -m & -\partial_t-\partial_x \\ \partial_t-\partial_x & -m \end{pmatrix} \psi = 0 \:. \]
This equation can be written in the Hamiltonian form
\[ i \partial_t \psi = H \psi \]
with
\[ 
H = \begin{pmatrix} i \partial_x & i m \\ -i m & -i \partial_x \end{pmatrix} \:. \]
This agrees with~\cite[eq.~(2.45)]{hollands-sanders}.

Employing the plane-wave ansatz
\[ \psi(t,x) = \f\: e^{-i \omega t + i k x} \]
with~$\omega, k \in \R$ and~$\f \in \C^2$, the Majorana equation becomes
\beq \label{majpw}
\begin{pmatrix} -m & i \omega - i k \\ -i \omega -i k & -m \end{pmatrix} \f = 0 \:.
\eeq
This equation has non-trivial solutions only if the determinant of the matrix vanishes. This gives the
usual dispersion relation
\[ \omega^2-k^2 = m^2 \:. \]

From now on, we restrict attention to the {\em{massive case}}~$m>0$
(the massless case can be obtained as a limiting case; see Remark~\ref{remm0}).
Then it is most convenient to parametrize the frequencies
and momenta as
\beq \label{parameter}
\begin{pmatrix} \omega \\ k \end{pmatrix} = ms \begin{pmatrix} \, \cosh \theta \\ -\sinh \theta \end{pmatrix} 
\qquad \text{with} \qquad s \in \{\pm 1\} \text{ and }\, \theta \in \R
\eeq
(here~$s$ is the sign of energy, and~$\theta$ describes the momentum).
Using these formulas in~\eqref{majpw} gives
\[ \begin{pmatrix} -m & i m s e^\theta \\ -i m s e^{-\theta} & -m \end{pmatrix} \f = 0 \:, \]
having the nontrivial solution
\beq \label{fdef}
\f = e^{-i \pi/4} \, \begin{pmatrix} e^{\theta/2} \\ -i s e^{-\theta/2} \end{pmatrix}\:.
\eeq
Here the phase factor~$e^{-i \pi/4}$ was introduced in order to get agreement with the 
conventions in~\cite{hollands-sanders}.

A general solution is obtained by taking superpositions of these plane-wave solutions. Taking
into account that the Majorana spinor is real-valued, we obtain the Fourier representation
\beq
\psi(t,x) = \int_{-\infty}^\infty d\theta \bigg( 
\begin{pmatrix} e^{\theta/2-i \pi/4} \\ e^{-\theta/2+i \pi/4} \end{pmatrix}\:
\overline{a(\theta)}\: e^{i \omega t - i k x}
+ \begin{pmatrix} e^{\theta/2+i \pi/4} \\ e^{-\theta/2-i \pi/4} \end{pmatrix}\:
a(\theta) \:e^{-i \omega t + i k x} \bigg) \:, \label{superpose}
\eeq
where~$\omega$ and~$k$ are given by~\eqref{parameter} with~$s=-1$.

\subsection{The Rindler wedge and the Rindler Hamiltonian}
The two-dimensional {\em{Rindler space-time}}~$(\scrR, g)$ is isometric to
the subset of two-dimensional Minkowski space
\[ 
\scrR = \big\{ (t,x) \in \R^{1,1} \;\;\;\text{with}\;\;\; |t| < x \big\} \]
with the induced line element again given by~\eqref{mink}.
We let~$S\scrR = \scrR \times \C^2$ be the trivial spinor bundle.

We consider Rindler spacetime as a subset of Minkowski space and also refer to it as the (right or left)
{\em{Rindler wedge}}. Then the inclusions
\[ \scrR \subset \scrM \qquad \text{and} \qquad S\scrR = \scrR \times \C^2
\subset \scrM \times \C^2 = S\scrM \]
are clearly isometries.
The Dirac operator and the Dirac equation are again given by~\eqref{Dirop}
and~\eqref{Direq}. In order to construct Dirac solutions in the Rindler wedge, one chooses
initial data~$\psi_0 \in C^\infty_0(\R^+, \C^2)$ and solves the Cauchy problem in~$\scrR$.
On these solutions we consider the scalar product
\[ 
(\psi | \phi)_\scrR = \int_0^\infty \Sl \psi | \gamma^0 \phi \Sr|_{(0,x)}\: dx \]
Forming the completion gives a Hilbert space denoted by~$(\H_\scrR, (.|.)_\scrR)$.
Extending the initial data~$\psi_0$ by zero to negative~$x$ and solving the Cauchy problem in
Minkowski space~$\scrM$, every solution in Rindler spacetime can be extended uniquely to
a solution in Minkowski space. We thus obtain the isometric embedding of Hilbert
spaces~$\H_\scrR \subset \H_\scrM$. We denote the orthogonal projection to the subspace~$\H_\scrR$
by
\beq \label{piR}
\pi_\scrR : \H_\scrM \rightarrow \H_\scrR \:.
\eeq

We now introduce the Rindler Hamiltonian as the infinitesimal generator of Lorentz boosts.
Recall that the {\em{Rindler coordinates}}~$(\rho, \tau)$ with~$\rho>0$ and~$\tau \in \R$
are defined by
\[ \begin{pmatrix} t \\ x \end{pmatrix} = \rho \begin{pmatrix} \sinh \tau \\ \cosh \tau \end{pmatrix} \:. \]
In these coordinates, the Rindler line element takes the form
\[ ds^2 = \rho^2\, d\tau^2 - d\rho^2 \:. \]
We work intrinsically in Rindler space-time. Translations in the time coordinate~$\tau$,
\[ 
\tau \mapsto \tau + \Delta \:,\qquad \rho \mapsto \rho \:, \]
describe a Killing symmetry. Therefore, writing the Dirac equation in this time coordinate in the Hamiltonian form
\beq \label{rindH}
i \partial_\tau \psi = H_\scrR \psi \:,
\eeq
the Dirac Hamiltonian is time independent (for details see the proof of~\cite[Theorem~10.1]{rindler}). 

\subsection{The Rindler vacuum and its excitations} \label{secexcite}
In order to describe the {\em{quantized Dirac field in two-dimensional Minkowski space}},
to every Dirac solution~$\psi \in \H_\scrM$ we associate fermionic field operators~$\Psi^\dagger(\psi)$
and~$\Psi(\overline{\psi})$ which satisfy the canonical anti-commutation relations~\eqref{CAR}
with corresponding scalar product~$\la .|. \ra = ( .|. )_\scrM$.
We denote the algebra generated by the fermionic field operators by~$\A_\scrM$.

Next, we choose~$\Omega$ as the vacuum in Minkowski space. To this end, we
decompose the Hilbert space~$\H_\scrM$ into the solutions of positive and negative frequency,
denoted by
\[ \H_\scrM = \H_+ \oplus \H_- \:. \]
Then~$\Omega$ is characterized by the fact that all negative-energy states are occupied, i.e.\
\beq \label{Omegadef}
\Psi(\overline{\psi})\, \Omega = 0 \quad \forall\; \psi \in \H_+ \qquad \text{and} \qquad
\Psi^\dagger(\psi)\, \Omega = 0 \quad \forall\;\psi \in \H_- \:.
\eeq
We choose a corresponding Fock representation and denote the resulting Fock space
by~$(\Fock, \la .|. \ra_\Fock)$.
We remark that occupying all states of {\em{negative}} energy is a convention
which implements the physical picture of the Dirac sea
(for basics see for example~\cite[Section~1.5]{intro}).
For the purposes of the present paper, one could occupy just as well
all positive-energy states.

In order to describe the {\em{quantized Dirac field in Rindler spacetime}}, we consider the subalgebra~$\A_\scrR$
generated by all the field operators~$\Psi^\dagger(\psi)$ and~$\Psi(\overline{\psi})$
with~$\psi \in \H_\scrR \subset \H_\scrM$. Restricting the vector~$\Omega \in \Fock$ to this subalgebra
defines the state
\[ 
\omega : \A_\scrR \rightarrow \C \:,\qquad \omega(A) = \la \Omega \:|\: A\, \Omega \ra_\Fock \:. \]
We consider this state as the {\em{ground state}} for the computation of the relative entropy.
A straightforward computation using the CAR shows that the state~$\omega$ is quasi-free and
particle number preserving. We denote its reduced one-particle density operator for notational
clarity by~$\sigma_\scrR$. According to~\eqref{omega2def} and~\eqref{Ddef}, it is defined by the
relation
\beq \label{sigmaRdef}
\la \psi | \sigma_\scrR \phi \ra
= \la \Omega \,|\, \Psi^\dagger(\phi)\, \Psi(\overline{\psi})\,\Omega \ra_\Fock \:,
\eeq
where for notational convenience we set~$\la .|. \ra = ( .|. )_\scrR = ( .|. )_\scrM |_{\H_\scrR \times \H_\scrR}$.
More specifically, using~\eqref{Omegadef},
\beq \label{sigmaRform}
\sigma_\scrR = \pi_\scrR \, \pi_- \, \pi_\scrR \::\: \H_\scrR \to \H_\scrR \:,
\eeq
where~$\pi_- \in \Lin(\H_\scrM)$ is the orthogonal projection operator to the negative-frequency solutions,
and~$\pi_\scrR$ is again the projection operator~\eqref{piR}.
For clarity, we point out that the state~$\omega$ is a KMS state. This corresponds to the fact
that the reduced one-particle density operator in~\eqref{sigmaRform} is not a projection operator.
This will become clearer when computing its spectral decomposition in Section~\ref{secrindlerred}.

Finally, the above setting simplifies if we restrict attention to the quantized {\em{Majorana field}}.
In this case, we only consider wave functions in~$\H_\scrR$ with real components, i.e.\
\[ \phi \in \H_\scrR \qquad \text{with} \qquad \phi(x) \in \R^2 \:. \]
The corresponding field operator should be a symmetric operator on the Fock space~$\Fock$.
We denote it for clarity by~$B(\phi)$. It can be expressed in terms of the Dirac field operators by
\[ B(\phi) = \frac{1}{\sqrt{2}}\: \big( \Psi^\dagger(\phi) + \Psi(\overline{\phi}) \big) \in \A_\scrR \:. \]
By direct computation using the CAR~\eqref{CAR} one verifies that the Majorana field operators satisfy the
relations
\begin{align}
    \label{bdef} B(\phi)^* = B(\phi) \qquad \text{and} \qquad \{ B(\phi), B(\psi) \} = \la \phi | \psi \ra = \la \psi | \phi \ra\end{align}
(the last relation holds because both~$\phi$ and~$\psi$ have real-valued components).

Our {\em{excited state}}~$\tilde{\omega}$ is introduced as follows. We let~$\phi \in \H_\scrR$ be a normalized
wave function (again with real-valued components).
We set~$\tilde{\Omega} = \sqrt{2}\, B(\phi)\, \Omega$ and
\begin{align}
    \label{tilomegadef}
\tilde{\omega} : \A_\scrR \rightarrow \C \:,\qquad 
\tilde{\omega}(A) = \la \tilde{\Omega} \:|\: A\, \tilde{\Omega} \ra_\Fock = 
2 \,\la \Omega \:|\: B(\phi)\, A\, B(\phi)\, \Omega \ra_\Fock \:. 
\end{align} 

We finally express the field operators in the familiar {\em{equal time formalism}}. To this end, we
write the field operators as
\[ \Psi^\dagger(\psi) = \sum_{\alpha=1}^2 \int_{-\infty}^\infty \psi^\alpha(0,x)\: \Psi^\alpha(x)^\dagger\: dx \]
and similarly for~$\Psi(\overline{\psi})$ (note that the Dirac solution~$\psi \in \H_\scrM$ is evaluated only at
time~$t=0$). A direct computation using~\eqref{CAR} and~\eqref{print} gives the CAR in the form
\[ \{ \Psi^\alpha(x), \Psi^\beta(y)^\dagger \} = \delta^{\alpha \beta}\: \delta(x-y)
\quad \text{and} \quad \{ \Psi^\alpha(x), \Psi^\beta(y) \} = 0 = \{ \Psi^\alpha(x)^\dagger, \Psi^\beta(y)^\dagger \} \:. \]
Likewise, the Majorana field operators~$B^\alpha(x)$ satisfy the relations
\begin{equation}
B^\alpha(x)^* = B^\alpha(x) \qquad \text{and} \qquad \{ B^\alpha(x), B^\beta(y) \}
= \delta^{\alpha \beta}\: \delta(x-y) \:.  \label{carb}
\end{equation} 
The field operators are solutions of the Dirac equation.
Therefore, in analogy to~\eqref{superpose}, the Majorana field operators can be expanded in plane waves as  
\[ B^\alpha(x) = \int_{-\infty}^\infty d\theta \bigg( 
\begin{pmatrix} e^{\theta/2-i \pi/4} \\ e^{-\theta/2+i \pi/4} \end{pmatrix}^\alpha\:
a(\theta)^\dagger\: e^{-i k x}
+ \begin{pmatrix} e^{\theta/2+i \pi/4} \\ e^{-\theta/2-i \pi/4} \end{pmatrix}^\alpha\:
a(\theta) \:e^{i k x} \bigg) 
\]
with~$\omega$ and~$k$ are again given by~\eqref{parameter} with~$s=-1$.
By direct computation
one verifies that the new field operators~$a(\theta)$ and~$a(\theta)^\dagger$
satisfy the CAR
\[ 
\{ a(\theta), a(\theta')^\dagger \} = \: \frac{m}{2 \pi} \:\delta(\theta - \theta')
\quad \text{and} \quad \{ a(\theta), a(\theta') \} = 0 = \{ a(\theta)^\dagger, a(\theta')^\dagger \} \:.  \]
Moreover, it is convenient to also express the smeared field operators~$B(\phi)$ in terms of the
operators~$a(\theta)$ and~$a(\theta)^\dagger$,
\beq \label{Bphi}
B(\phi) = \sum_{\alpha=1}^2 \int_{-\infty}^\infty \phi^\alpha(0,x)\: B^\alpha(x)\: dx 
= \int_{-\infty}^\infty d\theta \Big( f(\theta)\: a(\theta) + \overline{f(\theta)}\: a(\theta)^\dagger \Big) \:,
\eeq
where we set
\beq \label{ftheta}
f(\theta) := \sum_{\alpha=1}^2 \int_{-\infty}^\infty \phi^\alpha(0,x)\: 
\begin{pmatrix} e^{\theta/2+i \pi/4} \\ e^{-\theta/2-i \pi/4} \end{pmatrix}^\alpha\: e^{i k x} \: dx \:,
\eeq
which is an element of the one-particle Hilbert space $L^2(\R,d\theta)$.
\section{Computation via the reduced one-particle density operator} \label{seconeparticledensity}
\subsection{The reduced one-particle density operator in Rindler spacetime} \label{secrindlerred}
In this section we shall determine the spectral decomposition of the operator~$\sigma_\scrR$
as defined by~\eqref{sigmaRdef} (see Theorem~\ref{thmspec}). Our method is based on
similar computations in~\cite{rindler}. Indeed, the only difference of the present setup
to that in~\cite{rindler} is that, in order to work with real-valued Majorana spinors~\eqref{majorana},
we had to choose the Dirac matrices in the Majorana representation~\eqref{gammarep},
whereas in~\cite{rindler} the Dirac representation is used.
However, complexifying the Majorana spinors, we get back to Dirac solutions as considered in~\cite{rindler}.
We again restrict attention to the massive case~$m>0$
(the massless case will be treated in Remark~\ref{remm0} below).
Then the corresponding plane wave solutions
coincide with those constructed in~\eqref{parameter} and~\eqref{fdef}. In order to get a correspondence
with~\cite[eq.~(5.3)]{rindler}, we note that these plane wave solutions can be written as
\begin{align*} 
\psi(s,\theta) &= \f(s,\theta)\: e^{-i \omega t + i k x} \qquad \text{with} \\
\f(s,\theta) &= \frac{1}{\sqrt{2m}} \, \frac{e^{-i \pi/4}}{\sqrt{\epsilon(\omega)\, (\omega+k)}}\:
\begin{pmatrix} m \\ -i (\omega + k) \end{pmatrix} \\
&= \frac{1}{\sqrt{2m}} \, \frac{e^{-i \pi/4}}{\sqrt{m e^\theta}}\:
\begin{pmatrix} m \\ -i m s e^\theta \end{pmatrix}
= \frac{1}{\sqrt{2}} \, e^{-i \pi/4}\:
\begin{pmatrix} e^{-\theta/2} \\ -i s e^{\theta/2} \end{pmatrix}\:.
\end{align*}
Clearly, plane-wave solutions are not normalizable and
are therefore no vectors in a Hilbert space and no eigenvectors.
Nevertheless, we can compute the normalization integral in the
distributional sense. This procedure allows us to carry out the
necessary computations in a clear and convenient manner.
Moreover, we identify plane-wave solutions with generalized eigenfunctions (see for example the statement of Theorem~\ref{thmspec}).
In the computation of the relative entropy (Theorem~\ref{thmrelrindler})
we shall rewrite this normalization property in terms of a unitary operator on
a Hilbert space. The eigenvalues then correspond to points in the continuous
spectrum of a resulting multiplication operator. In this way, the results of our
computations will be restated in a clean way using functional analytic language.

\begin{Lemma}\label{lem51} The scalar products of the plane wave solutions in~$\scrM$ and~$\scrR$ are given
in the distributional sense by
\begin{align*}
\big( \psi(s,\theta) \:\big|\: \psi(\tilde{s}, \tilde{\theta}) \big)_\scrM &= 2 \pi \:\delta_{s,\tilde{s}}\: \delta \big( m\: \sinh \beta \big) \\
\big( \psi(s,\theta) \:\big|\: \psi(\tilde{s}, \tilde{\theta}) \big)_\scrR &= 
\pi \:\delta_{s,\tilde{s}}\: \delta \big( m\: \sinh \beta \big) + \frac{2i}{m}\: \frac{\text{\rm{PP}}}{e^{\beta} - s \tilde{s} e^{-\beta} }\:,
\end{align*}
where (as in~\cite[eq.~(6.3)]{rindler}) we set
\[ \beta := \frac{1}{2} \big( \theta - \tilde{\theta} \big) \:. \]
\end{Lemma}
\Proof The scalar product in~$\scrM$ can be written as
\begin{align*}
\big(\psi(s, \theta) \big| \psi(\tilde{s}, \tilde{\theta}) \big)_\scrM &= \int_{-\infty}^\infty \Sl 
\f(s, \theta) \,|\, \gamma^0 \,\f(\tilde{s}, \tilde{\theta}) \Sr\: e^{i (k-\tilde{k}) x}\:dx \\
&= 2 \pi \delta \big( k-\tilde{k} \big) \: \Sl \f(s, \theta) \,|\, \gamma^0 \,\f(\tilde{s}, \tilde{\theta}) \Sr \:.
\end{align*}
The terms in this formula are computed by
\begin{align}
\Sl \f(s, \theta) \,|\, \gamma^0 \,\f(\tilde{s}, \tilde{\theta}) \Sr
&= \frac{1}{2}\: \Big( e^{-\frac{1}{2}(\theta+\tilde{\theta})} + s \tilde{s} \,e^{\frac{1}{2} (\theta+\tilde{\theta})} \Big)
\label{iprod} \\
k - \tilde{k} &= m \big( s \sinh \theta - \tilde{s} \sinh \tilde{\theta} \big) \\
&= \frac{m}{2}\: \big( s e^\theta - s e^{-\theta} - \tilde{s} e^{\tilde{\theta}} + \tilde{s} e^{-\tilde{\theta}} \big) \\
&= \frac{m \tilde{s}}{2}\: \big( e^{-\frac{1}{2} (\theta+\tilde{\theta})}  + s \tilde{s} e^{\frac{1}{2}(\theta+\tilde{\theta})} \big)
\big( s e^{\frac{1}{2}(\theta-\tilde{\theta})} - \tilde{s} e^{\frac{1}{2} (\tilde{\theta}-\theta)} \big) \:. \label{deltaarg}
\end{align}
The argument of the $\delta$ distribution vanishes if one of the two factors in~\eqref{deltaarg} vanishes.
If the first factor vanishes, also~\eqref{iprod} is zero, so that we get no contribution to the scalar product.
Hence it suffices to take into account the zeros of the second factor. We thus obtain
\begin{align*}
\big(\psi(s, \theta) \big| \psi(\tilde{s}, \tilde{\theta}) \big)_\scrM
&= 4 \pi \: \delta \Big( m\: \big( \tilde{s}e^{\frac{1}{2}(\theta-\tilde{\theta})} -s e^{\frac{1}{2} (\tilde{\theta}-\theta)}
\big)\Big)
\: \epsilon \big( e^{-\frac{1}{2} (\theta+\tilde{\theta})}  + s \tilde{s} e^{\frac{1}{2}(\theta+\tilde{\theta})} \big)
\end{align*}
Clearly, the argument of the $\delta$ distribution vanishes only if~$s=\tilde{s}$, in which case
the argument of the step function is positive. Thus
\begin{align*}
\big(\psi(s, \theta) \big| \psi(\tilde{s}, \tilde{\theta}) \big)_\scrM
&= 4 \pi \:\delta_{s,\tilde{s}}\: \delta \big( 2 m\: \sinh \beta \big) 
= 2 \pi \:\delta_{s,\tilde{s}}\: \delta \big( m\: \sinh \beta \big) \:.
\end{align*}

In the Rindler spacetime, we integrate only over the half line,
\begin{align*}
\big(\psi(s, \theta) \big| \psi(\tilde{s}, \tilde{\theta}) \big)_\scrR &\!:= \int_{0}^\infty \Sl 
\f(s, \theta) \,|\, \gamma^0 \,\f(\tilde{s}, \tilde{\theta}) \Sr\: e^{i (k-\tilde{k}) x}\:dx \\
&= i \,\Big( \frac{\text{PP}}{k-\tilde{k}} - i \pi \delta \big(k-\tilde{k} \big) \Big)\:
\Sl \f(s, \theta) \,|\, \gamma^0 \,\f(\tilde{s}, \tilde{\theta}) \Sr \:.
\end{align*}
The contribution by the $\delta$ distribution was already computed above.
The principal value integral can be simplified as follows,
\[
i \:\frac{\text{PP}}{k-\tilde{k}} \:\Sl \f(s, \theta) \,|\, \gamma^0 \,\f(\tilde{s}, \tilde{\theta}) \Sr 
= \frac{2i}{m \tilde{s}}\: \frac{\text{PP}}{ \tilde{s} e^{\frac{1}{2}(\theta-\tilde{\theta})} -{s} e^{\frac{1}{2} (\tilde{\theta}-\theta)} } \:.
\]
Combining all the terms gives the result.
\QED

We point out that the scalar product computed in the previous lemma depends only
on the difference of the hyperbolic angles~$\theta$ and~$\tilde{\theta}$.
This corresponds to the symmetry of the considered spacetimes under Lorentz boosts.
This suggests that, similar as in~\cite[Section~7]{rindler}, it should be useful to form
a plane wave ansatz in the hyperbolic angle of the form
\beq \label{planewave}
\psi_\ell := \int_{-\infty}^\infty e^{-i \ell \theta}\: \psi(-1,\theta)\:d\theta \qquad \text{with~$\ell \in \R$}\:.
\eeq
The scalar products of these plane wave solutions can be computed with the usual $\delta$-normalization the ``boost momentum variable'' $\ell$.
\begin{Lemma}\label{lemplanewave} For the plane wave solutions~\eqref{planewave}
in the hyperbolic angle, the scalar products in~$\scrM$ and~$\scrR$ are given
in the distributional sense by
\begin{align}
(\psi_\ell \,|\,\psi_{\tilde{\ell}} )_\scrM &= 8 \pi^2\: \delta\big(m (\ell-\tilde{\ell}) \big) \label{pw1} \\
(\psi_\ell \,|\,\psi_{\tilde{\ell}} )_\scrR &= 4 \pi^2\: \delta\big(m (\ell-\tilde{\ell}) \big)
\Big( 1 - \tanh(2 \pi \ell) \Big) \:. \label{pw2}
\end{align}
\end{Lemma}
\Proof A direct computation gives
\begin{align*}
(\psi_\ell \,|\,\psi_{\tilde{\ell}} )_\scrM &= \int_{-\infty}^\infty d\theta \int_{-\infty}^\infty d\tilde{\theta}\:
e^{i \ell \theta - i \tilde{\ell} \tilde{\theta}} \: \big(\psi(-1,\theta) \,\big|\,\psi(-1,\tilde{\theta}) \big)_\scrM \\
&= \int_{-\infty}^\infty d\theta \int_{-\infty}^\infty d\tilde{\theta}\:
e^{i \ell \theta - i \tilde{\ell} \tilde{\theta}} \: 2 \pi\: \delta \Big( \frac{\theta - \tilde{\theta}}{2} \Big) \\
&= \frac{4 \pi}{m} \:\int_{-\infty}^\infty e^{i (\ell-\tilde{\ell}) \theta}\:d\theta = 8 \pi^2\: \delta\big(m (\ell-\tilde{\ell}) \big) \\
(\psi_\ell \,|\,\psi_{\tilde{\ell}} )_\scrR &- \frac{1}{2}\: (\psi_\ell \,|\,\psi_{\tilde{\ell}} )_\scrM
= \int_{-\infty}^\infty d\theta \int_{-\infty}^\infty d\tilde{\theta}\:
e^{i \ell \theta - i \tilde{\ell} \tilde{\theta}} \: \frac{2i}{m}\: \frac{\text{\rm{PP}}}{e^{\beta} - e^{-\beta} } \\
&= \int_{-\infty}^\infty d\theta \:e^{i (\ell-\tilde{\ell}) \theta}
\int_{-\infty}^\infty d\tilde{\theta}\:
e^{- i \tilde{\ell} (\tilde{\theta} - \theta)} \: \frac{2i}{m}\: \frac{\text{\rm{PP}}}{e^{\beta} - e^{-\beta} } \\
&= \int_{-\infty}^\infty d\theta \:e^{i (\ell-\tilde{\ell}) \theta}
\int_{-\infty}^\infty 2\,d\beta\:
e^{2 i \tilde{\ell}\beta} \: \frac{2i}{m}\: \frac{\text{\rm{PP}}}{e^{\beta} - e^{-\beta} } \\
&= \frac{8 \pi i}{m}\: \delta \big(\ell-\tilde{\ell} \big)
\int_{-\infty}^\infty e^{2 i \ell \beta} \:\frac{\text{\rm{PP}}}{e^{\beta} - e^{-\beta} }
\:d\beta \:.
\end{align*}
The last integral can be evaluated by contour integration. For $\ell > 0$, we close the contour in the upper half–plane, which requires taking into account the simple poles at $\beta = 0, i\pi, 2 i\pi, \ldots$, each with residue $\pm 2$. Because of the principal value prescription, the contribution of the pole at $\beta = 0$ is weighted by an additional factor of one half. We thus obtain
\begin{align*}
&\int_{-\infty}^\infty e^{- 2 i \ell \beta} \:\frac{\text{\rm{PP}}}{e^{\beta} - e^{-\beta} }\:d\beta
= i\pi  \:\bigg( \sum_{n=0}^\infty (-1)^n\: e^{-2 \ell \pi n} - \frac{1}{2} \bigg) \\
&= i \pi \:\bigg( \frac{1}{1+e^{-2 \pi \ell}} - \frac{1}{2} \bigg)
= \frac{i \pi}{2} \:\frac{1-e^{-2 \pi \ell}}{1+e^{-2 \pi \ell}} = \frac{i \pi}{2} \:\tanh(2 \pi \ell) \:.
\end{align*}
Similarly, in the case~$\ell<0$ we close the contour in the lower half plane to obtain
\begin{align*}
&\int_{-\infty}^\infty e^{- 2 i \ell \beta} \:\frac{\text{\rm{PP}}}{e^{\beta} - e^{-\beta} }\:d\beta
= -i \pi \:\bigg( \sum_{n=0}^\infty (-1)^n\: e^{2 \ell \pi n} - \frac{1}{2} \bigg) \\
&= -i \pi \:\bigg( \frac{1}{1+e^{2 \pi \ell}} - \frac{1}{2} \bigg)
= -\frac{i \pi}{2} \:\frac{1-e^{2 \pi \ell}}{1+e^{2 \pi \ell}} = \frac{i \pi}{2} \:\tanh(2 \pi \ell) \:.
\end{align*}
Combining all formulas gives the result.
\QED

From the formulas of this lemma, one immediately gets the following result.
\begin{Thm} \label{thmspec} Choosing~$\H=\H_-$ as all the negative-energy solutions, the
reduced one-particle density operator~$\sigma_\scrR$ of Rindler spacetime has an absolutely continuous
spectrum in~$[0,1]$. The corresponding generalized eigenfunctions and eigenvalues are given by
\begin{gather}
\sigma_\scrR \,(\pi_\scrR \psi_\ell) = \lambda_\ell\: (\pi_\scrR \psi_\ell) \qquad \text{with} \\
\lambda_\ell = \frac{1}{2} \:\Big( 1 - \tanh(2 \pi \ell) \Big) = \frac{e^{-2 \pi \ell}}{e^{2 \pi \ell} + e^{-2 \pi \ell}}
= \frac{1}{1 + e^{4 \pi \ell}}\:. \label{lamldef}
\end{gather}
\end{Thm}
\Proof According to~\eqref{pw1}, the projection operator~$\pi_-$ can be written as
\[ \pi_- = \frac{m}{8 \pi^2} \int_{-\infty}^\infty |\psi_\ell ) (\psi_\ell \,|\: d\ell \:. \]
A direct computation using~\eqref{pw2} yields
\begin{align*}
\pi_- \pi_\scrR \pi_-
&= \bigg( \frac{m}{8 \pi^2} \bigg)^2 \int_{-\infty}^\infty d\ell \int_{-\infty}^\infty d\ell' \;
|\psi_\ell )\; (\psi_\ell \,|\, \psi_{\tilde{\ell}} )_\scrR\; (\psi_{\tilde{\ell}} \,| \\
&= \frac{m}{8 \pi^2}
\int_{-\infty}^\infty \frac{1}{2}\:\Big( 1 - \tanh(2 \pi \ell) \Big) |\psi_\ell )(\psi_\ell \,|\: d\ell \:.
\end{align*}
Using the formula~\eqref{sigmaRform} for~$\sigma_\scrR$, it follows, again using~\eqref{pw2}, that
\[ \sigma_\scrR \:\pi_\scrR \psi_\ell =  \pi_\scrR \,  \pi_- \, \pi_\scrR   \,\psi_\ell 
= \lambda_\ell\, \pi_\scrR \,\psi_\ell \]
with~$\lambda_\ell$ according to~\eqref{lamldef}. This gives the result.
\QED

As is shown in detail in~\cite[Section~10]{rindler}, the solutions of the form~\eqref{planewave} are also eigenfunctions of the Rindler Hamiltonian~\eqref{rindH}. More precisely,
\beq \label{HRspec}
H_\scrR \psi_\ell = - \ell \psi_\ell \:.
\eeq
Therefore, we can write the result of the above Theorem as
\beq \label{sigmaR}
\sigma_\scrR = \big( 1 + e^{-4 \pi H_\scrR} \big)^{-1} \:.
\eeq
This coincides with the operator~$Q(\beta)$ in~\cite{galanda2023relative} with inverse temperature
given by
\[ \beta = 4 \pi \:. \]

We finally explain how the above results can be generalized to the massless case.
\begin{Remark} \label{remm0} {\em{ {\bf{(massless Dirac particles)}}
The massless case~$m=0$ can be obtained from the above analysis by a suitable limiting case
obtained as follows. Clearly, taking the naive limit~$m \searrow 0$ in the ansatz~\eqref{parameter}
is not sensible because the right side tends to zero. This problem can be avoided by
also letting~$\theta=\theta(m)$ tend to plus or minus infinity. More precisely, we take the limits such that
\[ \lim_{m \searrow 0} \frac{m}{2} e^{\pm \theta_\pm(m)} = p \qquad \text{with} \qquad p>0 \:, \]
because then~\eqref{parameter} reduces to
\[ \begin{pmatrix} \omega \\ k \end{pmatrix} = s \begin{pmatrix} p \\ \pm p \end{pmatrix} \]
(thus~$p$ is the absolute value of the momentum, whereas~$\pm$ determines its sign).
Using this method, we can also take the limit of all other formulas. For example,
in the case of plus signs,
\[  m\: \sinh \beta = \frac{m}{2}\: \Big( e^{\frac{\theta_+-\tilde{\theta}_+}{2}} - e^{\frac{-\theta_+ +\tilde{\theta}_+}{2}} \Big)
= e^{\frac{-\theta_+ - \tilde{\theta}_+}{2}} \: \Big( \frac{m}{2}\: e^{\theta_+} - \frac{m}{2}\: e^{\tilde{\theta}_+} \Big)\:. \]
After rescaling, this converges to~$p-\tilde{p}$. The other cases can be treated similarly.
The plane wave ansatz in the rapidity variable~\eqref{planewave} needs to be modified according to
\[ e^{-i \ell \theta} =  \exp \Big\{ \mp i \ell \log \Big( \frac{2p}{m} \Big) \Big\} 
= e^{\mp i \ell \log(2/m)}\: \exp \Big\{ \mp i \ell \log p \Big\} \:. \]
After dropping the irrelevant phase factor~$e^{\mp i \ell \log(2/m)}$, we can take the limit~$m \searrow 0$.
In this way, one sees that Theorem~\ref{thmspec} also holds in the massless case.
}} \QEDrem
\end{Remark}

\subsection{The relative entropy of vacuum excitations} \label{secrel1}
Our goal is to compute the relative entropy of the vacuum excitations introduced in Section~\ref{secexcite}.
Before entering the analysis, we give a brief outline of the our methods.
In~\cite[Theorem~A.7]{fermientropy} a formula was derived which expresses the relative entropy
of quasi-free fermionic states in terms of their reduced one-particle density operators.
This result has a shortcoming that it applies only if both the vacuum and the excited state
are {\em{particle-number preserving}}, meaning that
all two-point expectations involving two creation or two annihilation operators vanish,
\beq \label{particle-perserve}
\omega\big( \Psi^\dagger(\phi)\, \Psi^\dagger(\psi) \big) = 0 =  \omega\big( \Psi(\overline{\phi})\, \Psi(\overline{\psi}) \big) \:.
\eeq
This condition is satisfied for the vacuum state of Rindler spacetime.
However, it is {\em{not}} satisfied for the excited state.
For this reason, here we cannot apply~\cite[Theorem~A.7]{fermientropy}.
Instead, there are two alternative methods. The first method worked out in Appendix~\ref{appA}
is to compute the relative entropy in the Fock space using that the vacuum is a particle-number preserving
quasi-free state of the form considered in~\cite[Lemma~A.2]{fermientropy}.
The second, more general method worked out in Appendix~\ref{appB} is to extend the
result~\cite[Theorem~A.7]{fermientropy} to general Gaussian states, i.e.\ quasi-free states
which do {\em{not}} preserve the particle number.

We now state and derive a formula for the relative entropy based on the method in Appendix~\ref{appA}.
The method in Appendix~\ref{appB} will be explored in the next subsection (Section~\ref{secaltgauss}).
In Appendix~\ref{appA} it is shown in the finite-dimensional setting
how the relative entropy can be expressed in terms of the reduced one-particle density operator
of the vacuum and the form of excitation. The resulting formula from Theorem~\ref{thmexcite}
also applies in the infinite-dimensional setting if one simply replaces sums by convergent series,
giving the following result.
\begin{Thm} \label{thmrelrindler}
The relative entropy of the Rindler vacuum~$\omega$ and the excited state~$\tilde{\omega}$
can be expressed by
\[ S(\tilde{\omega}\Vert\omega) = 4 \pi \: \big\la f \:\big|\: H_\scrR \tanh \big(2 \pi H_\scrR \big) f \big\ra_{\H } \:. \]
\end{Thm}
\Proof In preparation, we restate the results of the above computations
in a functional analytic language. A direct computation using~\eqref{pw1}
shows that the operator~$U$ defined by
\[ U \::\: L^2(\R, d\ell) \rightarrow \H_\scrM = \H_- \:,\quad
U f := \sqrt{\frac{m}{8 \pi^2}} \int_{-\infty}^\infty f(\ell)\: \psi_\ell\: d\ell \]
is isometric. Moreover, since the plane-wave solutions~$\psi(-1,\theta)$ span~$\H_-$,
it follows that it is even unitary. Now we can restate the result of Theorem~\ref{thmspec}
\[ U^{-1} \sigma_\scrR U = T_h \:, \]
where~$T_h$ is the multiplication operator
\[ T_h \::\: L^2(\R, d\ell) \rightarrow L^2(\R, d\ell)\:,\quad
(T_h f)(\ell) = h(\ell)\, f(\ell) \qquad \text{with} \qquad
h(\ell) = \frac{1}{1 + e^{4 \pi \ell}} \:. \]

We now want to apply Theorem~\ref{thmexcite} in Appendix~\ref{appA}.
This theorem was derived in finite dimensions. But by a standard approximation
argument, it also holds in separable Hilbert spaces, stating that
\beq \label{Srelform}
S(\tilde{\sigma}\Vert\sigma_\scrR) = - \la f \:|\: S\, (2D-1)\: f \ra_\H
= - ( f \:|\: S\, (2D-1)\: f )_\scrM
\eeq
(in the last step we use that~$\H_\scrR$ is a subspace of~$\H_\scrM$).
By unitary invariance, the last expectation value can be computed on~$L^2(\R, d\ell)$,
\beq \label{continuous}
S(\tilde{\sigma}\Vert\sigma_\scrR) 
= -\int_{-\infty}^\infty |f_\ell|^2\: s_\ell\: (2 d_\ell - 1)\:d\ell \:.
\eeq
Using the notation from Appendix~\ref{appA}, we have
\begin{align*}
d_\ell &= \lambda_\ell = \frac{1}{1 + e^{-4 \pi \ell}} \\
1-d_\ell &= \frac{e^{-4 \pi \ell}}{1 + e^{-4 \pi \ell}} \\
s_\ell &= \log \Big( \frac{1-d_\ell}{d_\ell} \Big) = \log \big( e^{-4 \pi \ell} \big) = - 4 \pi \ell \\
1-2d_\ell &= \frac{e^{-4 \pi \ell} - 1}{1 + e^{-4 \pi \ell}} = - \frac{1 - e^{-4 \pi \ell}}{1 + e^{-4 \pi \ell}}
= - \frac{e^{2 \pi \ell} - e^{-2 \pi \ell}}{e^{2 \pi \ell}  + e^{-2 \pi \ell}} = - \tanh ( 2 \pi \ell) \:.
\end{align*}
Employing these formulas in~\eqref{continuous}, we obtain
\[ S(\tilde{\sigma}\Vert\sigma_\scrR)
= 4 \pi \int_{-\infty}^\infty |f_\ell|^2\, \ell\, \tanh(2 \pi \ell) \:d\ell\:. \]
Using~\eqref{HRspec} gives the result.
\QED

\subsection{Alternative computation via general Gaussian states} \label{secaltgauss}
We now explain how the relative entropy can be computed using the formula
for the relative entropy for general Gaussian states derived in Appendix~\ref{appB}
(see Theorem~\ref{thmgauss}). Clearly, we get the same result as already obtained
in Theorem~\ref{thmrelrindler}. Nevertheless, it is instructive to see how the different methods
and formulas fit together.

According to~\eqref{bogu} and~\eqref{Udef},
\[ \tilde{\Psi}_n = f_n \,\big( \Psi^\dagger(f) + \Psi(\overline{f}) \big) - \Psi_n \:, \]
where we reverted the roles of the tilde and untilde operators.
Hence the unitary transformation~$\scrU$ in~\eqref{scrUdef} takes the form
\beq \label{Ubog}
\scrU = -\1 + \begin{pmatrix} |f \ra \la f| & |f \ra \la \overline{f}| \\ | \overline{f} \ra \la f| & |\overline{f} \ra \la \overline{f}| \end{pmatrix}\:.
\eeq
We transform to the normal modes of the untilde operators. Then, using the notation~\eqref{Tprp},
\begin{align*}
T &= \begin{pmatrix} D & 0 \\ 0 & \1-D \end{pmatrix} =: \frac{e^{-\mathscr{C}}}{2 \cosh \mathscr{C}} \qquad
\text{and thus} \\
\1-T &= \begin{pmatrix} \1-D & 0 \\ 0 & D \end{pmatrix} = \frac{e^{\mathscr{C}}}{2 \cosh \mathscr{C}}
\:,\qquad \frac{T}{\1-T} = e^{-2 \mathscr{C}} \:.
\end{align*}
We conclude that
\[ {\mathscr{C}} = -\frac{1}{2} \: \log \Big( \frac{T}{\1-T} \Big)
= \frac{1}{2}\: \begin{pmatrix} S & 0 \\ 0 & -S \end{pmatrix} 
\qquad \text{with} \qquad S :=  \log \Big( \frac{\1-D}{D} \Big) \:. \]
Next,
\begin{align*}
&\tilde{T}-T = \scrU \,T\, \scrU^{-1} -T \\
&= - \begin{pmatrix} |f \ra \la f| & |f \ra \la \overline{f}| \\ | \overline{f} \ra \la f| & |\overline{f} \ra \la \overline{f}| \end{pmatrix} \: T 
- T\: \begin{pmatrix} |f \ra \la f| & |f \ra \la \overline{f}| \\ | \overline{f} \ra \la f| & |\overline{f} \ra \la \overline{f}| \end{pmatrix} \\
&\quad\: + \begin{pmatrix} |f \ra \la f| & |f \ra \la \overline{f}| \\ | \overline{f} \ra \la f| & |\overline{f} \ra \la \overline{f}| \end{pmatrix}\,T\, \begin{pmatrix} |f \ra \la f| & |f \ra \la \overline{f}| \\ | \overline{f} \ra \la f| & |\overline{f} \ra \la \overline{f}| \end{pmatrix} \\
&= \begin{pmatrix} -|f \ra \la f|\,D & |f \ra \la \overline{f}|\,D \\ -| \overline{f} \ra \la f|\,D & |\overline{f} \ra \la \overline{f}| \,D \end{pmatrix}
- \begin{pmatrix} 0 & |f \ra \la \overline{f}| \\ 0 & |\overline{f} \ra \la \overline{f}| \end{pmatrix} \\
&\quad\:+ \begin{pmatrix} -D\,|f \ra \la f| & -D\,|f \ra \la \overline{f}| \\ D\,| \overline{f} \ra \la f| & D\,|\overline{f} \ra \la \overline{f}| \end{pmatrix} 
- \begin{pmatrix} 0 & 0 \\ | \overline{f} \ra \la f| & |\overline{f} \ra \la \overline{f}| \end{pmatrix}
\\
&\quad\: + \begin{pmatrix} 0 & |f \ra \la \overline{f}| \\ 0 & |\overline{f} \ra \la \overline{f}| \end{pmatrix}\, \begin{pmatrix} |f \ra \la f| & |f \ra \la \overline{f}| \\ | \overline{f} \ra \la f| & |\overline{f} \ra \la \overline{f}| \end{pmatrix} \\
&= \begin{pmatrix} -|f \ra \la f|\,D & |f \ra \la \overline{f}|\,D \\ -| \overline{f} \ra \la f|\,D & |\overline{f} \ra \la \overline{f}| \,D \end{pmatrix}
- \begin{pmatrix} 0 & |f \ra \la \overline{f}| \\ 0 & |\overline{f} \ra \la \overline{f}| \end{pmatrix} \\
&\quad\:+ \begin{pmatrix} -D\,|f \ra \la f| & -D\,|f \ra \la \overline{f}| \\ D\,| \overline{f} \ra \la f| & D\,|\overline{f} \ra \la \overline{f}| \end{pmatrix} 
- \begin{pmatrix} 0 & 0 \\ | \overline{f} \ra \la f| & |\overline{f} \ra \la \overline{f}| \end{pmatrix} \\
&\quad\: + \begin{pmatrix} |f \ra \la f| & |f \ra \la \overline{f}| \\ | \overline{f} \ra \la f| & |\overline{f} \ra \la \overline{f}| \end{pmatrix} \\
&= \begin{pmatrix} -|f \ra \la f|\,D & |f \ra \la \overline{f}|\,D \\ -| \overline{f} \ra \la f|\,D & |\overline{f} \ra \la \overline{f}| \,D \end{pmatrix} + \begin{pmatrix} -D\,|f \ra \la f| & -D\,|f \ra \la \overline{f}| \\ D\,| \overline{f} \ra \la f| & D\,|\overline{f} \ra \la \overline{f}| \end{pmatrix} + \begin{pmatrix} |f \ra \la f| &0 \\ 0 & -|\overline{f} \ra \la \overline{f}| \end{pmatrix} \:.
\end{align*}
Employing these formulas in the relation~\eqref{Srel2} in Theorem~\ref{thmgauss}, we obtain
\begin{align*}
S(\tilde{W}\Vert W) &= \tr_\H \Big\{ (\tilde{T} - T) \:\mathscr{C} \Big\} \\
&= \frac{1}{2} \: \tr_\H \Big\{ (\tilde{T} - T) \:\begin{pmatrix} S & 0 \\ 0 & -S \end{pmatrix} \Big\} \\
&= \frac{1}{2} \: \tr_\H \Big( -4\, |f \ra \la f|\,D\,S + 2 \,|\overline{f} \ra \la \overline{f}| \,S \Big) \\
&= \la f \,|\, (1-2 D) S \:f \ra \:,
\end{align*}
and this coincides precisely with~\eqref{Srelform}.

\section{Computation using modular theory} \label{secmodular}
 
In order to calculate the relative entropy using modular theory, we use the framework of the self-dual CAR algebra introduced by~\cite{araki1970quasifree} and explicit formulas given in~\cite{galanda2023relative}. In the following, we concisely summarize the paper's content in a manner tailored to the application; for more details, please refer to the original papers. 

First, let $\H$ be a complex Hilbert space with scalar product 
$\braket{\cdot\vert\cdot}_{\H}$, together with an anti-linear operator $\Gamma : \H \to
\H$ which is anti-unitary, so that 
\[ \Gamma^2 = {\bf 1} \quad \text{and} \quad 
    \braket{\Gamma {f}\vert\Gamma {g}}_{\H}= \braket{{g}\vert{f}}_{\H} \quad \ \
    ({f},{g} \in \H) \]

\begin{Def}
{\rm~\cite{araki1970quasifree} } \ \ {\tt CAR}$(\H,\Gamma)$, the $C^*$-algebra of the {\rm self-dual canonical anti-commu\-tation 
relations} with respect to $(\H,\Gamma)$, is defined as the (up to $C^*$-equivalence unique)
unital $C^*$-algebra generated by the unit element ${\bf 1}$ and elements $B({f})$, ${f} \in 
\H$, fulfilling the conditions
\begin{gather*}
{f} \mapsto B({f})  \ \ \text{is complex-linear}\,, \quad B({f})^* = B(\Gamma {f})\,, \quad \text{and}\\
[B({f})^*, B({g})]_+  = \braket{{f}\vert{g}}_{\H} {\bf 1}
\end{gather*}
for all ${f},{g} \in \H$, where $[X,Y]_+ = XY + YX$ denotes the algebraic
{\it anti-commutator bracket}.
\end{Def}
\begin{Def} \label{def:stand}
    Let $\omega$ be a standard state on {\tt CAR}$(\H,\Gamma)$ and let $\Omega$ be its vector representative.
\\[6pt]
(a) \ \ For any function ${f}$ with the following properties
\begin{enumerate}
    \item $\Gamma{f} = {f}$ and 
    \item $\braket{{f}\vert{{f}}}_{\H} = 2$
\end{enumerate}
  the CAR-generator \( B({f}) \) is unitary and we call the state vector $F\Omega$, where $F = \pi_\omega(B({f}))$, 
a  \emph{single-excitation state} relative to $\omega$. This state is denoted by $\omega_{f}$.
\\[6pt] 
(b) \ \ We call a standard state $\omega$ a {\bf{standard 1-particle flow state (S1PFS)}} if there 
is a 1-particle flow $\{u_t\}_{t \in \mathbb{R}}$ on $(\H,\Gamma)$ (i.e.\ a 1-parametric
unitary group with the property $\Gamma u_t=u_t\Gamma$) such that 
\[ 
\Delta^{it} \pi_\omega(B({f})) \Delta^{-it} = \pi_\omega \big( B(u_t{f}) \big) \]
holds for all ${f} \in \H$ and all $t \in \mathbb{R}$. 
    
\end{Def}

We define the {\it 2-point function}  w.r.t.\ any state on {\tt CAR}$(\H,\Gamma)$ by the following relation 
\begin{align*}
 W^{(2)}_\omega({f},{g}) &= \omega(B({f})B({g}))
 \\&= \braket{{\sf \Gamma f}\vert{Q_\omega{g}}}_{\H}\quad \ \ ({f},{g} \in \H)\,.
\end{align*}
where $Q_\omega : \H \to \H$ is a linear operator,  called the {\it base polarization} of $\omega$,
characterized by the following properties (cf.~\cite{araki1970quasifree}):
\begin{align*}
 0 \le Q_\omega = &\ Q^*_\omega \le {\bf 1}\,,  \quad Q_\omega + \Gamma Q_\omega \Gamma = {\bf 1}\,, \quad \text{and} \\
 & W^{(2)}_\omega(\Gamma{f},{g}) = \braket{{\sf \Gamma f}\vert{Q_\omega{g}}}_{\H}.
\end{align*}
   For a S1PFS $\omega$, one can explicitly state the form
of $Q_\omega$. To this end, we assume that a one-particle space flow \(\{u_t\}_{t \in \mathbb{R}}\) on \(\H\) is given. The flow is expressed as  
\[
u_t = {\rm e}^{-itH},   
\]  
where \(H\) is a self-adjoint operator defined on a dense domain \({\rm dom}(H)\) in \(\H\) and for \(\beta > 0\), we then define  
\begin{align} \label{eq:Qbeta}  
Q_{(\beta)} = ({\bf 1} + {\rm e}^{-\beta H})^{-1}.
\end{align}
If the function \( {f} \in \H \) satisfies \( \Gamma {f} = {f} \) and has the norm $\Vert f\Vert^2_{\H}= 2$, then the CAR-generator \( B({f}) \) is unitary. This follows directly from the defining relations of the CAR-generators. 

Building on the previous results and definitions, we now introduce the key quantity of interest, namely the relative entropy.
\begin{Thm}\label{thms1}
 Let $\omega$ be a {\rm S1PFS} on {\tt CAR}$(\H,\Gamma)$, where 
 $\{u_t\}_{t \in \mathbb{R}}$ with $u_t = {\rm e}^{-itH}$ denotes the associated one-particle flow. Then, for any single-excitation state $\omega_{f}$\footnote{See Definition \ref{def:stand}.} relative to $\omega$, the relative entropy is 
 \begin{align} \label{eq:1exent}
  S(\omega_{f} \Vert \omega) & = i
    \left. \frac{d}{dt} \right|_{t = 0} W^{(2)}_\omega({f},u_t{f}) 
    \\
    &= 4\pi  \braket{{f}\vert{Q_{(4\pi)}H{f}}}_{\H} \label{eq:2exent}
\end{align}
where $Q_{(4\pi)}$ is given by~\eqref{eq:Qbeta} for $\beta=4\pi$.
Thus, the relative entropy $S(\omega_{f} \Vert \omega)$ is finite if 
${f}$ is contained in the 
domain of $|H|^{1/2}$.
\end{Thm} \label{thm:AandB} 
Next, we calculate the relative entropy explicitly for the case of the Rindler spacetime. In this particular example, the Hamiltonian is given by the boost operator in the $0-1$ direction, denoted by $H_{\scrR}$. It generates Lorentz boosts with matrices $\Lambda_t$ of the following form
\[ 
 \Lambda_t  =  \left(\begin{matrix} \cosh{t} &\sinh{t} \\\sinh{t}  & \cosh{t}\end{matrix}\right). \]
The main expression we are interested in, is given according to Theorem~\ref{thms1}  by
\begin{equation}\label{eqrelentrop}
    S(\omega_f \Vert \omega) = 4\pi\braket{{f}\vert{Q_{(4\pi)}H_{\scrR}{f}}}_{\H} 
\end{equation}
where  $Q_{(4\pi)}=\left(1+e^{-4\pi H_{\scrR}}\right)^{-1}$ and $H_{\scrR}$ is the Rindler-Hamiltonian. The Rindler Hamiltonian can be deduced from the energy momentum tensor that is in the flat case given by
\begin{align*}
    \Theta^{\mu\nu}=\frac{i}{2} \bar{\Psi}\gamma^{\mu}\overleftrightarrow{\partial^{\nu}}\Psi
\end{align*} 
Next, we calculate the relative entropy for an explicit example. In particular, we take, following~\cite[Section 2.4]{hollands-sanders}, the Majorana field in $1+1$ dimensional Minkowski spacetime. The Hilbert space becomes in this example is  $L^2(\mathbb{R} ,dx;\mathbb{C}^2)$ and the inner product is the standard inner product on this space. The Dirac field operator in the representation is then given by the self-adjoint (since bounded and symmetric) $2$-component operator valued distribution that is given as by
\begin{align*}
 B^{\alpha}(x)=  \frac{1}{\sqrt{2}}\int_{\mathbb{R}}d\theta \left\{
\left(\begin{matrix} e^{\theta/2-i\pi/4}  \\ e^{-\theta/2+i\pi/4}  \end{matrix}\right)^{\alpha}
    e^{-ip(\theta)x}a(\theta)^{\dagger}+ \left(\begin{matrix} e^{\theta/2+i\pi/4}  \\ e^{-\theta/2-i\pi/4}  \end{matrix}\right)^{\alpha}
    e^{+ip(\theta)x}a(\theta)  \right\} ,
\end{align*}
where $x=(x_0,x_1)$,  $p(\theta)=(-m\cosh{\theta}, m\sinh{\theta})$ and $a(\theta), a^{\dagger}(\theta)$ fulfill anti-commutation relations. Note that the (infinitesimal) boosted\footnote{As per standard convention, we use the symbol $S(\Lambda)$ to denote the boosts in the spinor case.} spinor is given by
\[ i \frac{d}{dt}\bigg|_{t = 0} S(\Lambda_t)B^{\alpha}(\Lambda_{-t}y)= \Big( \frac{1}{2}\:\sigma^{01}
-i \big(y^0\partial^1+y^1\partial^0 \big) \Big)\,B^{\alpha}(y)=H_{\scrR}B^{\alpha}(y) \:. \]
where by a slight abuse of notation we use the same symbol $H_{\scrR}$ for the one-particle and second-quantized operator.
Next, we define, equivalently to~\eqref{Bphi}, \eqref{ftheta}), the smeared Majorana field  operator $B(\phi)$\footnote{Note we denote the representation of the operator $B$ and the operator with the same symbol.}  
\begin{align}\label{bfi}
    B(\phi)&=\sum_{\alpha=1}^2 \int \phi^{\alpha}(0,x)B^{\alpha}(x) \,dx
\\&=\int_{-\infty}^\infty d\theta \Big( f(\theta)\: a(\theta) + \overline{f(\theta)}\: a(\theta)^\dagger \Big)
\end{align}
It is easy to check that this satisfies the relation 
\begin{align*}
    \{B(\phi) ,B(\psi)\}&=\frac{m}{2\pi}\int d\theta \left(f(\theta)\overline{g(\theta)}+\overline{f(\theta)}g(\theta)\right)
   =\langle \phi\vert \psi\rangle 
\end{align*}
where we used Equations~\eqref{ftheta},~\eqref{carb} and used the notation 

$$g(\theta) := \sum_{\alpha=1}^2 \int_{-\infty}^\infty \psi^\alpha(0,x)\: 
\begin{pmatrix} e^{\theta/2+i \pi/4} \\ e^{-\theta/2-i \pi/4} \end{pmatrix}^\alpha\: e^{i k x} \: dx \:.$$   
Next, we calculate the action of the Rindler Hamiltonian $H_{\scrR}$ on the operator $B(\phi)$ using its action on the operator $B^{\alpha}$.
\begin{Lemma} let $B (\phi)$ denote the operator  given in Equation~\eqref{bfi}.  Then, the action of the Rindler Hamiltonian $H_{\scrR}$ on the operator $B(\phi)$ is explicitly given by
\[ H_{\scrR}  B(\phi)=     \int_{-\infty}^\infty \Big( f(\theta)\: i\frac{\partial}{\partial \theta}a(\theta)
+ \overline{f(\theta)}\: \big(-i\frac{\partial}{\partial \theta}a(\theta) \big)^\dagger \Big)\: d\theta \:. \]

\end{Lemma} 
\begin{proof}
The action of the Rindler Hamiltonian $H_{\scrR}$ on the operator $B^{\alpha}$ is  
\[ H_{\scrR} B^{\alpha}(x)= \Big( \frac{1}{2}\:\sigma^{01}-i \big(x^0\partial^1+x^1\partial^0 \big) \Big)\:B^{\alpha}(x) \:, \]
where the matrix $\sigma^{01}$ is defined by
\[ \sigma^{01}=\frac{1}{2}[\gamma^0,\gamma^1] :=
\begin{pmatrix}- i &0 \\0
& i\end{pmatrix}   \:. \]
We begin by evaluating the first part of the Rindler operator, rendering 
\begin{align*}&\frac{1}{2}\ \sigma^{01}B^{\alpha}(x)  \\&=
 \frac{1}{2\sqrt{2}}\int_{\mathbb{R}}d\theta \left\{
  \sigma^{01}\left(\begin{matrix} e^{\theta/2-i\pi/4}  \\ e^{-\theta/2+i\pi/4}  \end{matrix}\right)^{\alpha}
    e^{-ip(\theta)x}a(\theta)^{\dagger}+  \sigma^{01} \left(\begin{matrix} e^{\theta/2+i\pi/4}  \\ e^{-\theta/2-i\pi/4}  \end{matrix}\right)^{\alpha}
    e^{+ip(\theta)x}a(\theta)  \right\}  \\
    &= \frac{1}{2\sqrt{2}}\int_{\mathbb{R}}d\theta \left\{
 \left(\begin{matrix} -i\, e^{\theta/2-i\pi/4}  \\  i\, e^{-\theta/2+i\pi/4}  \end{matrix}\right)^{\alpha}
    e^{-ip(\theta)x}a(\theta)^{\dagger}+   \left(\begin{matrix}  -i\, e^{\theta/2+i\pi/4}  \\   i\, e^{-\theta/2-i\pi/4}  \end{matrix}\right)^{\alpha}
    e^{+ip(\theta)x}a(\theta)  \right\} 
\end{align*}
For the  second part of the Rindler Hamiltonian $H_{\scrR}$, we note that we have the relation 
\begin{align*}
-i(x^0\partial^1+x^1\partial^0) \:e^{-ip(\theta)x}&=- \big(x^0 p^1(\theta) -x^1p^0(\theta) \big)\:e^{-ip(\theta)x}
\\&=-
i\frac{\partial}{\partial \theta} e^{-ip(\theta)x}
\end{align*}
and, analogously,
\begin{align*}
-i(x^0\partial^1+x^1\partial^0) e^{ ip(\theta)x}&= \big(x^0 p^1(\theta) -x^1p^0(\theta) \big)\:e^{ ip(\theta)x}
\\&=-
i\frac{\partial}{\partial \theta} e^{ip(\theta)x} \:.
\end{align*}
Thus, acting with the second part of the Rindler Hamiltonian $H_{\scrR}$ on the operator $B^{\alpha}$,
we obtain
\begin{align*}&-i(x^0\partial^1+x^1\partial^0) B^{\alpha}(x) \\&=
 \frac{1}{ \sqrt{2}}\int_{\mathbb{R}}d\theta \left\{
   \left(\begin{matrix} e^{\theta/2-i\pi/4}  \\ e^{-\theta/2+i\pi/4}  \end{matrix}\right)^{\alpha}
  \left(- i\frac{\partial}{\partial \theta} e^{-ip(\theta)x}\right)a(\theta)^{\dagger} \right. \\
  &\qquad\qquad\qquad \left. +  \left(\begin{matrix} e^{\theta/2+i\pi/4}  \\ e^{-\theta/2-i\pi/4}  \end{matrix}\right)^{\alpha}
    \left( -i\frac{\partial}{\partial \theta} e^{+ip(\theta)x}\right)a(\theta)  \right\} \:.
\end{align*}
Integration by parts gives
\begin{align*}&
    -i(x^0\partial^1+x^1\partial^0) B^{\alpha}(x)=-\frac{1}{2}\ \sigma^{01}B^{\alpha}(x) \\&+
     \frac{1}{ \sqrt{2}}\int_{\mathbb{R}}d\theta \left\{
   \left(\begin{matrix} e^{\theta/2-i\pi/4}  \\ e^{-\theta/2+i\pi/4}  \end{matrix}\right)^{\alpha}
     e^{-ip(\theta)x}\:  i\frac{\partial}{\partial \theta}a(\theta)^{\dagger}+  \left(\begin{matrix} e^{\theta/2+i\pi/4}  \\ e^{-\theta/2-i\pi/4}  \end{matrix}\right)^{\alpha}   e^{ ip(\theta)x}\:  
 i\frac{\partial}{\partial \theta} a(\theta)  \right\} .
\end{align*}
After smearing the vector and distribution-valued operator $B_{\alpha}$, we obtain 
\begin{align*}
    \sum_{\alpha=1}^2 \int \phi^{\alpha}(0,x)H_{\scrR}B^{\alpha}(x) \,d x=
    \int_{-\infty}^\infty \Big( f(\theta)\: i\frac{\partial}{\partial \theta}a(\theta) + \overline{f(\theta)}\: (-i\frac{\partial}{\partial \theta}a(\theta))^\dagger \Big)\:d\theta \:,
\end{align*}
giving the result.
\end{proof}  
\begin{Thm}\label{lemmalb}
    Let $\omega$ be a S1PFS state, and let $B(\phi)$ denote the operator  given in Equation~\eqref{bfi}. Then, the relative entropy
$$S(\omega_f \Vert \omega)=4\pi \:\langle f \vert Q_{\omega} H_{\scrR}f \rangle_{\H} $$
(see~\eqref{eq:2exent} and~\eqref{eqrelentrop}) is given explicitly by 
   \begin{align*}
    S(\omega_f \Vert \omega)    &=2\pi\: \big\la f \:\big|\: H_\scrR \tanh \big(2 \pi H_\scrR \big) f \big\ra_{\H },
    \end{align*}
    which agrees with the relative entropy given in Theorem~\ref{thmrelrindler} if we choose the excited state $\omega_f=\tilde{\omega}$.
\end{Thm}

\begin{proof}
We compute the expectation value $\langle f \vert Q_{\omega} H_{\scrR} f \rangle_{\H}$ by passing to the Fourier conjugate variable, where the modular Hamiltonian acts as a multiplication operator. 

Recall that the modular group $\Delta^{it}$ of the Rindler wedge acts on the one-particle space $L^2(\R, d\theta)$ as translations in the rapidity variable $\theta$:
\begin{align*}
    (\Delta^{it} f)(\theta) = f(\theta - 2\pi t).
\end{align*}
It follows that the Rindler Hamiltonian $H_\scrR$, defined by $\Delta^{it} = e^{-2\pi i t H_\scrR}$, acts as the generator of these translations, i.e.\ $H_\scrR = -i\partial_\theta$ on suitable domains. Introducing the Fourier transform
\begin{align*}
    \hat{f}(\ell) = \frac{1}{2\pi} \int_{-\infty}^{\infty} e^{i\ell\theta}\, f(\theta)\, d\theta,
\end{align*}
$H_\scrR$ becomes the multiplication operator $\ell \mapsto \ell$ on $L^2(\R, d\ell)$. By the spectral theorem and the Borel functional calculus for self-adjoint operators, the bounded operator $Q_\omega = \frac{1}{2}\big(1 + \tanh(2\pi H_\scrR)\big)$ likewise becomes a multiplication operator in Fourier space:
\begin{align*}
    \widehat{(Q_\omega g)}(\ell) = \frac{1}{2}\big(1 + \tanh(2\pi \ell)\big)\, \hat{g}(\ell).
\end{align*}
Here we have used the identity $(1 + e^{-4\pi \ell})^{-1} = \frac{1}{2}(1 + \tanh(2\pi \ell))$ relating the operator representation in the former Lemma to the hyperbolic tangent form.

Using Parseval's theorem and the fact that both $Q_\omega$ and $H_\scrR$ are multiplication operators in the $\ell$-variable, we obtain
\begin{align*}
    \langle f \vert Q_{\omega} H_{\scrR} f \rangle_{\H}
    &= \int \hat{\bar{f}}(\ell)\; \frac{1}{2}\big(1 + \tanh(2\pi \ell)\big)\; \ell\; \hat{f}(\ell)\, d\ell \\
    &= \frac{1}{2} \int \ell\, |\hat{f}(\ell)|^2\, \big(1 + \tanh(2\pi \ell)\big)\, d\ell.
\end{align*}
Splitting the integrand, the term $\frac{1}{2}\int \ell\, |\hat{f}(\ell)|^2\, d\ell$ vanishes by the parity symmetry $\ell \mapsto -\ell$, since $\ell\, |\hat{f}(\ell)|^2$ is odd when integrated against the even measure $d\ell$.\footnote{More precisely, for $f \in C_0^\infty(\R^2,\C^2)$, we have $|\hat{f}(\ell)|^2 = |\hat{f}(-\ell)|^2$ by the reality-type symmetry of the test function space, so the odd part integrates to zero. Alternatively, note that $\langle f | H_\scrR f\rangle$ is real for self-adjoint $H_\scrR$ and simultaneously purely imaginary by the skew-symmetry of $-i\partial_\theta$; hence it vanishes.} We are left with
\begin{align*}
    \langle f \vert Q_{\omega} H_{\scrR} f \rangle_{\H}
    &= \frac{1}{2} \int \ell\, |\hat{f}(\ell)|^2\, \tanh(2\pi \ell)\, d\ell.
\end{align*}
Multiplying by $4\pi$, the relative entropy becomes
\begin{align*}
    S(\omega_f \Vert \omega) 
    &= 2\pi \int \ell\, |\hat{f}(\ell)|^2\, \tanh(2\pi \ell)\, d\ell \\
    &= 2\pi\, \big\la f \:\big|\: H_\scrR \tanh\big(2\pi H_\scrR\big)\, f \big\ra_{\H},
\end{align*}
where the last equality follows from the functional calculus, since $H_\scrR \tanh(2\pi H_\scrR)$ is the multiplication operator $\ell \mapsto \ell\, \tanh(2\pi \ell)$ in Fourier space. Note that $\ell \tanh(2\pi\ell) \geq 0$ for all $\ell \in \R$, confirming positivity of the relative entropy.

Choosing the state $\widetilde{\omega} = 2\omega_f$ (where the factor of $2$ resolves the discrepancy in~\eqref{tilomegadef}), the quantum relative entropy is given by
\begin{align*}
    S(\tilde{\omega} \Vert \omega)
 =4\pi\: \big\la f \:\big|\: H_\scrR \tanh \big(4\pi H_\scrR \big) f \big\ra_{\H },
\end{align*}
which matches the result derived in Theorem~\ref{thmrelrindler}.
\end{proof}

\section{Non-unitary excitations of more general states} \label{secbeyond}
With the previous constructions we exemplified that all computations 
performed previously with the help of modular theory can also be carried out
using the reduced one-particle density operator, giving the same results.
We now go one step further and show in simple examples that the computational methods
using the reduced one-particle density can be used even in cases when modular theory does not apply.
To this end, we consider more general ground states where the one-particle density operator~$\sigma_\scrR$
has the form
\[ \sigma_\scrR = \eta(H_\scrR) \]
with~$\eta \::\: \R \rightarrow [0,1]$ a Borel function. 
Choosing~$\eta(\ell) = \big( 1 + e^{-4 \pi \ell} \big)^{-1}$, we get back the state considered in~\eqref{sigmaR}.
But now more general states are possible, like for example
\begin{align*}
\eta(\ell) &= \big( 1 + e^{\beta \ell} \big)^{-1} && \text{general inverse temperature~$\beta \in (0, \infty)$} \\
\eta(\ell) &= \Theta(-\ell) && \text{zero temperature} \\
\eta(\ell) &= \frac{1}{2} && \text{infinite temperature} \:.
\end{align*}
The zero temperature state corresponds to the ground state of an observer at infinity.
All states are all quasi-free and particle-number preserving.
As the excitations we can consider arbitrary quasi-free states (not necessarily particle-number
preserving). The corresponding relative entropy can be computed using the formula
in Theorem~\ref{thmgauss}.

As a simple example, let us assume that the operator~$\sigma_\scrR:= \eta(H_\scrR)$ has an
eigenvalue~$\lambda \in (0,1)$.
The existence of an eigenvalue can be arranged by choosing~$\eta$ in such a way that the set
\[ \text{$\eta^{-1}(\lambda) \subset \R$ has strictly positive Lebesgue measure}\:. \]
We let~$\f \in \H_\scrR$ be a corresponding eigenvector~$f$, i.e.\
\[ \sigma_\scrR f = \lambda f \:. \]
As the simplest example, one can consider the
infinite temperature state with~$\lambda=\mu=1/2$.
We let~$\omega$ be the quasi-free, particle-number preserving state corresponding to~$\sigma_\scrR$.
We introduce the excited state by
\beq \label{exnonunit}
\omega_f(\cdots) := \frac{1}{Z}\,\omega \Big( \big( 1 + \Psi(\overline{f}) + \Psi^\dagger(f) \big)
\cdots \big(1 + \Psi^\dagger(f) + \Psi(\overline{f}) \big) \Big)
\eeq
with the normalization constant
\[ Z = \omega \Big( \big( 1 + \Psi(\overline{f}) + \Psi^\dagger(f) \big)
\big(1 + \Psi^\dagger(f) + \Psi(\overline{f}) \big) \Big) = 1+|f|^2\:. \]
In the limiting case~$|f|\rightarrow \infty$ is normalized, the excitation can be described similar to~\eqref{Ubog} 
and~\eqref{Udef} by a Bogoliubov transformation. However, this is {\em{not}} possible otherwise.
Therefore, the relative entropy cannot be computed via modular theory. But, working with the reduced one-particle density operators, the relative entropy can indeed be described explicitly.
\begin{Prp} The relative entropy of the excitation~\eqref{exnonunit} is given by
\beq \label{Srelf2}
S(\omega_f \Vert \omega) = -\tr_{\C^2} \Big\{ T \big( \log T - \log T_0 \big) \Big\} \:,
\eeq
where~$T$ and~$T_0$ are the $2 \times 2$-matrices
\beq \label{TT0}
T = \frac{1}{1+|f|^2}
\begin{pmatrix} \lambda + (1-\lambda)\: |f|^2 & 0 \\ 0 & (1-\lambda) + \lambda |f|^2
\end{pmatrix} \:,\qquad
T_0 = \begin{pmatrix} \lambda & 0 \\ 0 & 1-\lambda \end{pmatrix}\:.
\eeq
\end{Prp}
\Proof We again use the notation and results of Appendix~\ref{appB}. 
It clearly suffices to consider the corresponding operators~$T_0$ and~$T$
on the two-dimensional space spanned by~$f$ and~$\bar{f}$. We represent the restrictions
onto this two-dimensional subspace by $2 \times 2$-matrices.
Then the operator~$T_0$ in Theorem~\ref{thmgauss} takes the form as in~\eqref{TT0}.
Moreover, from the relation~\eqref{Tprp} in Proposition~\ref{prpgauss}, it follows that
\beq \label{C0ex}
\mathscr{C}_0 = -\frac{1}{2} \: \log \Big( \frac{T_0}{1-T_0} \Big)
= \begin{pmatrix} c_0 & 0 \\ 0 & -c_0 \end{pmatrix}
\qquad \text{with} \qquad c_0 := -\frac{1}{2} \log \Big( \frac{\lambda}{1-\lambda} \Big) \:.
\eeq

Next, we need to compute the reduced one-particle density operator~$T$
of the excited state. Using~\eqref{CAR}, we obtain
\begin{align*}
&Z\: \omega_f \big( \Psi^\dagger_i \Psi_j \big) =
\omega \Big( \big(1+\Psi(\overline{f}) + \Psi^\dagger(f) \big) \:\Psi^\dagger_i \Psi_j\:
\big(1+\Psi^\dagger(f) + \Psi(\overline{f}) \big) \Big) \\
&= \omega \Big( \big(1+\Psi(\overline{f}) + \Psi^\dagger(f) \big) \:\Psi^\dagger_i\, f_j \Big) 
-\omega \Big( \big(1+\Psi(\overline{f}) + \Psi^\dagger(f) \big) \:\overline{f}_i \Psi_j \big) \Big)\\
&\quad\: + \omega \Big( \big(1+\Psi(\overline{f}) + \Psi^\dagger(f) \big)^2\:\Psi^\dagger_i \Psi_j\: \Big) \\
&= \omega \Big( \Psi(\overline{f}) \:\Psi^\dagger_i\, f_j \Big) 
-\omega \Big( \Psi^\dagger(f)\:\overline{f}_i \Psi_j \big) \Big) + (1+|f|^2)\: \omega \Big(\Psi^\dagger_i \Psi_j\: \Big) \\
&= (1-\lambda)\: \overline{f}_i f_j  - \overline{f}_i f_j\: \lambda + \big(1 + |f|^2 \big)\: \lambda \,\delta_{ij} \\
&Z\: \omega_f \big( \Psi^\dagger_i \Psi^\dagger_j \big)  =
\omega \Big( \Psi(\overline{f}) \:\Psi^\dagger_i \Psi^\dagger_j\:\Psi(\overline{f}) \Big) \\
&= \overline{f_j} \:\omega \Big( \Psi(\overline{f}) \:\Psi^\dagger_i \Big)
- \omega \Big( \Psi(\overline{f}) \:\Psi^\dagger_i\:\Psi(\overline{f})\: \Psi^\dagger_j\ \Big) \\
&= \overline{f_j} \:\omega \Big( \Psi(\overline{f}) \:\Psi^\dagger_i \Big)
- \overline{f}_i \:\omega \Big( \Psi(\overline{f}) \:\Psi^\dagger_j\ \Big)
= -\overline{f_i} \overline{f_j} \: \big( \lambda - \lambda \big) = 0 \\
&Z\: \omega_f \big( \Psi_i \Psi_j \big) =
\omega \Big( \Psi^\dagger(f) \:\Psi_i \Psi_j\:\Psi^\dagger(f) \Big) \\
&= f_j \:\omega \Big( \Psi^\dagger(f) \:\Psi_i \Big)
- \omega \Big( \Psi^\dagger(f) \:\Psi_i\:\Psi^\dagger(f)\: \Psi_j\ \Big) \\
&= f_j \:\omega \Big( \Psi^\dagger(f) \:\Psi_i \Big)
- f_i\: \omega \Big( \Psi^\dagger(f) \: \Psi_j\ \Big) = 0 \\ 
&Z\: \omega_f \big( \Psi_i \Psi^\dagger_j \big) =
\omega \Big( \big(1+\Psi(\overline{f}) + \Psi^\dagger(f) \big) \:\Psi_i \Psi^\dagger_j\:
\big(1+\Psi^\dagger(f) + \Psi(\overline{f}) \big) \Big) \\
&=- f_i\: \omega \Big( \big(1+\Psi(\overline{f}) + \Psi^\dagger(f) \big) \:\Psi^\dagger_j \Big) 
+ \overline{f_j}\: \omega \Big( \big(1+\Psi(\overline{f}) + \Psi^\dagger(f) \big) \:\Psi_i \Big) \\
&\quad\: +\omega \Big( \big(1+\Psi(\overline{f}) + \Psi^\dagger(f) \big)^2 \:\Psi_i \Psi^\dagger_j \Big) \\
&=- f_i\: \omega \Big( \Psi(\overline{f})\:\Psi^\dagger_j \Big) 
+ \overline{f_j}\: \omega \Big(\Psi^\dagger(f) \:\Psi_i \Big)
+ \big(1+|f|^2 \big)\: \omega \Big( \Psi_i \Psi^\dagger_j \Big) \\
&=- (1-\lambda)\: \overline{f_j} f_i+ \overline{f_j} f_i\: \lambda
+ (1+|f|^2) \: (1-\lambda)\: \delta_{ij} \:.
\end{align*}
Hence, again on the subspace spanned by~$f$ and~$\overline{f}$,
\begin{align}
T &= \begin{pmatrix} \lambda & 0 \\ 0 & 1-\lambda
\end{pmatrix} + \frac{(1-2 \lambda)\:|f|^2}{1+|f|^2} \begin{pmatrix} 1 & 0 \\ 0 & -1 \end{pmatrix} \notag \\
&= \frac{1}{1+|f|^2}
\begin{pmatrix} \lambda + (1-\lambda)\: |f|^2 & 0 \\ 0 & 1-\lambda + \lambda |f|^2
\end{pmatrix} \\
T-T_0 &= \frac{(1-2 \lambda)\:|f|^2}{1+|f|^2} \begin{pmatrix} 1 & 0 \\ 0 & -1 \end{pmatrix} \label{TmT0} \:.
\end{align}
This gives the formula for~$T$ in~\eqref{TT0}.

We now employ these formulas in~\eqref{Srel0}. According to~\eqref{C0ex},
the operator~$\cosh \mathscr{C}_0$ is a multiple of the identity matrix.
Moreover, the relation~\eqref{TmT0} shows that the operator~$T-T_0$ is trace-free. Therefore,
\[ \tr_\H \Big\{ \big( T -T_0 \big) \,\log (2 \cosh \mathscr{C}_0) \Big\} = 0 \:. \]
With this in mind, the relation~\eqref{Srel0} simplifies to~\eqref{Srelf2}.
\QED

\section{Comparison of the methods and outlook} \label{seccompare}
In this paper, we explored two different approaches for computing relative entropies:
modular theory and computations based on the reduced one-particle density operator.
A-priori, these approaches apply in different setting and under different assumptions.
The starting point of modular theory is a Neumann algebra~$\mathcal{M}$ represented on a
Hilbert space together with a cyclic separating vector~$\Omega$.
The existence of this cyclic vector is quite restrictive. In simple terms, it applies only to
bipartite quantum systems, where one subsystem is generated by~${\mathcal{M}}$, whereas
the other is generated by its commutant~$\mathcal{M}'$.
In this setting, the Araki-Uhlmann relative entropy is defined for states which are unitarily related.
Despite these restrictions, modular theory is very general. In particular,
the resulting states~$\omega$ and $\tilde{\omega}$ do not necessarily need to be quasi-free.

For the approach based on the reduced one-particle operator, on the other hand, one does not
need a bipartite system. Moreover, the states considered for the relative entropy do not need
to be unitarily related. Instead, the method applies to arbitrary quasi-free states and an arbitrary 
quasi-free excitations. But, of course, the methods do not extend to fully interacting 
(i.e.\ non-quasi-free) states.

In Sections~\ref{seconeparticledensity} and~\ref{secmodular} we analyzed
excitations in the Rindler spacetime where both methods apply.
We verified by detailed direct computation that the results agree.
Apart from being an important consistency check, this analysis also
explains why and how these methods fit together.
In Section~\ref{secbeyond} we gave a simple example of an excitation which
lies outside the realm of modular theory, but where the relative entropy can still be
computed explicitly by means of the reduced one-particle density operator.
We expect that other physically interesting quantum systems and excitations can
be analyzed similarly by adapting our methods.
It seems an interesting open problem to analyze the relative entropy for non-quasi-free states using
modular theory.

\appendix
\section{The relative entropy of vacuum excitations} \label{appA}
In this appendix we compute the relative entropy of vacuum excitations.
We use the same setup and notation as in~\cite[Appendix~A]{fermientropy}.
Thus we assume that the one-particle Hilbert space~$\H$ is finite-dimensional of dimension~$N$.
The corresponding fermionic Fock space~$\F$ has dimension~$2^N$.
The fermionic field operators, denoted by~$\Psi_n^\dagger$ and~$\Psi_n$ with~$n=1,\ldots, N$,
act on this Fock space and realize the canonical anti-commutation relations~\eqref{CAR}

We assume that the vacuum state~$W_0$ is Gaussian and {\em{particle-number preserving}}
(as defined in~\eqref{particle-perserve}).
Then, in a suitable basis of~$\H$, the vacuum state can be written as
\beq \label{W0def}
W_0 = \det(\1-D_0) \:\exp \bigg( -\sum_{n=1}^N s_n\, \Psi_n^\dagger \Psi_n \bigg) \:,
\eeq
where~$D_0$ is the matrix
\[ D_0 = \text{diag} \big( d_1, \ldots, d_N ) \]
with real eigenvalues~$d_n \in (0,1)$ and
\[ 
s_n := \log \Big( \frac{1-d_n}{d_n} \Big) \]
(the cases~$d_n=0$ and~$d_n=1$ can be described as limiting cases where~$s_n \rightarrow \pm \infty$).

We describe the excitation by a transformation~$U$ on the Fock space~$\F$ of the specific form
\beq \label{Udef}
U = \Psi^\dagger(f) + \Psi(\overline{f}) \qquad \text{with} \qquad \la f,f\ra_{\C^N}=1\:.
\eeq
The operator~$U$ is obviously symmetric and unitary. Transforming the field operators according to
\[ \tilde{\Psi}_n := U \Psi_n U^* \:,\qquad \tilde{\Psi}^\dagger_n =  U \Psi^\dagger_n U^* \:, \]
the excited state is described by the density operator
\beq \label{Wdef}
W = \det(\1-D_0) \:\exp \bigg( -\sum_{n=1}^N s_n\, \tilde{\Psi}_n^\dagger \tilde{\Psi}_n \bigg) \:.
\eeq
\begin{Prp} \label{thmexcite}
The relative entropy between the vacuum state~$W_0$ and the excited state~$W$
(given by~\eqref{W0def} and~\eqref{Wdef}) can be expressed by
\[ S(W \Vert W_0) := \tr_\Fock \big(W \:(\log W-\log W_0) \big) = -\sum_{k=1}^N |f_k|^2\, s_k \:\big( 2d_k - 1 \big) \:. \]
\end{Prp}
\Proof Taking the logarithm of~\eqref{W0def} and using~\eqref{Wdef}, one finds
\begin{align*}
&\tr_\Fock \big(W \log W_0) = \log \det(\1-D_0) \\
&\qquad\: - \det(\1-D_0) \: \tr_\Fock \bigg\{
\exp \bigg( -\sum_{n=1}^N s_n\, \tilde{\Psi}_n^\dagger \tilde{\Psi}_n \bigg)
\sum_{k=1}^N s_k\:\Psi_k^\dagger \Psi_k \bigg\} \:.
\end{align*}
We now use that
\begin{align}
\tilde{\Psi}^\dagger(f) &= U \Psi^\dagger(f) U = 
\big( \Psi^\dagger(f) + \Psi(\overline{f}) \big)\, \Psi^\dagger(f)\, \big( \Psi^\dagger(f) + \Psi(\overline{f}) \big) \notag \\
&= \Psi(\overline{f})\, \Psi^\dagger(f)\, \Psi(\overline{f}) = \Psi(\overline{f}) \\
\tilde{\Psi}(\overline{f}) &= \Psi^\dagger(f) \\
U &= \tilde{\Psi}^\dagger(f) + \tilde{\Psi}(\overline{f}) \label{Utildef}
\end{align}
and thus
\begin{align}
\Psi_n &= U \tilde{\Psi}_n U^* = U \tilde{\Psi}_n U = U\, \big\{ \tilde{\Psi}_n, U \big\} - U U \tilde{\Psi}_n = f_n\, U - \tilde{\Psi}_n \label{bogu} \\
\Psi^\dagger_n &:=  U \tilde{\Psi}^\dagger_n U^* = U \tilde{\Psi}^\dagger_n U = U\, \big\{ \tilde{\Psi}^\dagger_n, U \big\}
- U U \tilde{\Psi}^\dagger_n = \overline{f_n} \, U  - \tilde{\Psi}^\dagger_n \\
\Psi^\dagger_n \Psi_n 
&= \tilde{\Psi}^\dagger_n \tilde{\Psi}_n + |f_n|^2 - \tilde{\Psi}^\dagger_n \, f_n\, U -  \overline{f_n} \, U\, \tilde{\Psi}_n \:.
\end{align}
Hence
\begin{align*}
&\tr_\Fock \big(W \log W_0) - \log \det(\1-D_0) \\
&= -\det(\1-D_0) \: \tr_\Fock \bigg\{ e^{-\sum_{n=1}^N s_n\, \tilde{\Psi}^\dagger_n \tilde{\Psi}_n}
\sum_{k=1}^N s_k\:\Big( \tilde{\Psi}^\dagger_k \tilde{\Psi}_k + |f_k|^2 - \tilde{\Psi}^\dagger_k \, f_k\, U -  \overline{f_k} \, U\, \tilde{\Psi}_k \Big) \bigg\} \:.
\end{align*}
Here the trace is zero unless, for every~$\ell \in \{1, \ldots, N\}$, the number of creation operators~$\Psi^\dagger_\ell$
equals the number of annihilation operators~$\Psi_\ell$. We thus obtain
\begin{align*}
&\tr_\Fock \big(W \log W) - \log \det(\1-D_0) \\
&= -\det(\1-D_0) \: \tr_\Fock \bigg\{ e^{-\sum_{n=1}^N s_n\, \tilde{\Psi}^\dagger_n \tilde{\Psi}_n}
\sum_{k=1}^N s_k\: \tilde{\Psi}^\dagger_k \tilde{\Psi}_k \bigg\} \\
&= -\det(\1-D_0) \: \sum_{k=1}^N s_k \prod_{n \neq k} \frac{1}{1-d_n} 
\tr_{\C^2} \bigg\{ \begin{pmatrix} 1 & 0 \\ 0 & e^{-s_k} \end{pmatrix}
\begin{pmatrix} 0 & 0 \\ 0 & 1 \end{pmatrix} \bigg\} \\
&= -\sum_{k=1}^N s_k \:(1-d_k)\: e^{-s_k}
= -\sum_{k=1}^N s_k \:(1-d_k)\: \frac{d_k}{1-d_k}
=  -\sum_{k=1}^N s_k \:d_k \\
&=  -\sum_{k=1}^N d_k\: \log \Big( \frac{1-d_k}{d_k} \Big)
= -\sum_{k=1}^N d_k\: \log (1-d_k) + \sum_{k=1}^N d_k\: \log(d_k) \\
&\tr_\Fock \big(W \log W_0) - \log \det(\1-D_0) \\
&= -\det(\1-D_0) \: \tr_\Fock \bigg\{ e^{-\sum_{n=1}^N s_n\, \tilde{\Psi}^\dagger_n \tilde{\Psi}_n}
\sum_{k=1}^N s_k\,\Big( \big(1 - 2\, |f_k|^2 \big)\, \tilde{\Psi}^\dagger_k \tilde{\Psi}_k + |f_k|^2 \Big) \bigg\} \\
&= -\det(\1-D_0) \: \sum_{k=1}^N s_k \prod_{n \neq k} \frac{1}{1-d_n} 
\tr_{\C^2} \bigg\{ \begin{pmatrix} 1 & 0 \\ 0 & e^{-s_k} \end{pmatrix}
\begin{pmatrix} |f_k|^2 & 0 \\ 0 & 1- |f_k|^2 \end{pmatrix} \bigg\} \\
&= -\sum_{k=1}^N s_k \:(1-d_k)
\tr_{\C^2} \bigg\{ \begin{pmatrix} 1 & 0 \\ 0 & e^{-s_k} \end{pmatrix}
\begin{pmatrix} |f_k|^2 & 0 \\ 0 & 1- |f_k|^2 \end{pmatrix} \bigg\} \\
&= -\sum_{k=1}^N s_k \:d_k \Big( e^{s_k}\: |f_k|^2 + 1 - |f_k|^2 \Big) \:.
\end{align*}

We conclude that
\begin{align*}
S(W \Vert W_0) 
&=\sum_{k=1}^N s_k \:d_k \Big( e^{s_k}\: |f_k|^2 + 1 - |f_k|^2 \Big) - \sum_{k=1}^N s_k \:d_k \\
&=\sum_{k=1}^N |f_k|^2\, s_k \:d_k \Big( e^{s_k} -1 \Big) \\
&=\sum_{k=1}^N |f_k|^2\, s_k \:d_k \Big( \frac{1-d_k}{d_k} -1 \Big)
=\sum_{k=1}^N |f_k|^2\, s_k \:(1-2 d_k)
\:.
\end{align*}
This gives the result.
\QED

\section{Fermionic entropies of general Gaussian states} \label{appB}
In the appendix in~\cite{fermientropy} it was shown how to express fermionic entropies in terms
of the reduced one-particle density operator. All the computations were carried out under the
assumption that the Gaussian state is {\em{particle-number preserving}}, meaning that
all two-point expectations involving two creation or two annihilation operators vanish,
\[  \omega\big( \Psi^\dagger(\phi)\, \Psi^\dagger(\psi) \big) = 0 =  \omega\big( \Psi(\overline{\phi})\, \Psi(\overline{\psi}) \big) \:. \]
The state obtained by exciting the Rindler vacuum, however, is {\em{not}} particle-number preserving.
This makes it necessary to extend the formulas derived in~\cite[Appendix~A]{fermientropy} to
general Gaussian states. This will be worked out in detail in this appendix.
We remark that an alternative method for treating Gaussian states which do not preserve the
particle number is to transform to suitable normal modes. This method is explained
in~\cite[Section~8]{fermientropy} in the context of fermionic lattices.

We use the same setup and notation as in~\cite[Appendix~A]{fermientropy}
as summarized at the beginning of Appendix~\ref{appA} above.

\subsection{Bogoliubov transformations}
A {\em{Bogoliubov transformation}} is a transformation of the field operators of the general form
\[ \tilde{\Psi}_n = \sum_{k=1}^N \big( u_{nk}\, \Psi_m + v_{nk}\, \Psi^\dagger_k \big) \:, \]
where the coefficients~$u_{nk}$ and~$v_{nk}$ must be chosen such that the CAR remain valid.
It is convenient to write the transformation of these operators and their adjoints in matrix form as
\beq \label{scrUdef}
\begin{pmatrix} \tilde{\Psi} \\ \tilde{\Psi}^\dagger \end{pmatrix} 
= \scrU  \begin{pmatrix} \Psi \\ \Psi^\dagger \end{pmatrix} \qquad \text{with} \qquad \scrU := 
\begin{pmatrix} U & V \\ \overline{V} & \overline{U} \end{pmatrix}\:,
\eeq
where the bar denotes component-wise complex conjugation of all matrix elements.
A direct computation shows that, in order to preserve the CAR, the matrices~$U$ and~$V$
must satisfy the conditions
\[  \scrU \scrU^* = \1 \:. \]
Thus the matrix~$\scrU$ must be unitary. Moreover, the specific form of~$\scrU$ in~\eqref{scrUdef}
can be characterized by the relation
\beq \label{SUrel}
S \scrU S = \overline{\scrU} \qquad \text{with} \qquad S := \begin{pmatrix} 0 & 1 \\ 1 & 0 \end{pmatrix}\:.
\eeq

\subsection{General Gaussian states}
We consider the ansatz
\beq \label{Wdefgen}
W := \frac{1}{\tr_\Fock \exp(-Q)}\: \exp (-Q) \:,
\eeq
where~$Q$ is the most general symmetric quadratic functional in the fermionic field operators, which we write as
\beq \label{Qform}
Q := a_{jk} \Psi^\dagger_j \Psi_k + b_{jk} \Psi_j \Psi_k + \overline{b_{kj}}\:\Psi^\dagger_j \Psi^\dagger_k
- a_{kj} \Psi_j \Psi^\dagger_k \:,
\eeq
where
\[ \overline{a_{jk}} = a_{kj} \qquad \text{and} \qquad b_{jk} = - b_{kj} \:. \]
Thus, in matrix notation~$A=(a_{jk})$ and~$B=(b_{jk})$,
\beq \label{ABrel}
A^* = A \qquad \text{and} \qquad B^T = -B \:.
\eeq
Note that the first relation is needed in order for~$Q$ to be a symmetric operator on~$\Fock$,
whereas the second relation is no loss in generality in view of the CAR.
Likewise, the form of the last summand in~\eqref{Qform} can be arranged using the CAR.

We write~$Q$ more compactly as
\beq \label{QCdef}
Q = \bigg\langle \begin{pmatrix} \Psi \\ \Psi^\dagger \end{pmatrix} , \mathscr{C}
\begin{pmatrix} \Psi \\ \Psi^\dagger \end{pmatrix} \bigg\rangle 
\qquad \text{with} \qquad \mathscr{C} := 
\begin{pmatrix} A & B \\ -\overline{B} & -\overline{A} \end{pmatrix} \:.
\eeq
The specific form of~$\mathscr{C}$ is characterized by the relation
\beq \label{SCrel}
S \mathscr{C} S = -\overline{\mathscr{C}} \:.
\eeq

The matrix~$\mathscr{C}$ is Hermitian (as one sees from~\eqref{QCdef} and~\eqref{ABrel}).
Therefore, it can be diagonalized with the help of a unitary transformation~$\scrU$, i.e.\
\beq \label{Cdiag}
\scrU \mathscr{C} \scrU^{-1} = \tilde{\mathscr{C}}
\eeq
with~$\tilde{\mathscr{C}}$ a diagonal $2n \times 2n$-matrix. It turns out that the unitary matrix can be chosen
such as to satisfy the relation~\eqref{SUrel}. In other words, the matrix~$\mathscr{C}$ can be
diagonalized by a Bogoliubov transformation. For completeness, we give a proof of this well-known fact.
\begin{Lemma} There is a unitary transformation~$\scrU$ satisfying~\eqref{SUrel} such that~\eqref{Cdiag}
holds. Moreover, the matrix~$\tilde{\mathscr{C}}$ is of the form
\beq \label{tildeCdef}
\tilde{\mathscr{C}} = \begin{pmatrix} C & 0 \\ 0 & -C \end{pmatrix} \:,
\eeq
where~$C$ is a diagonal $n \times n$-matrix with non-negative entries.
\end{Lemma}
\Proof Let~$\psi$ be an eigenvector of~$\mathscr{C}$, i.e.\
\[ \mathscr{C} \psi = \lambda \psi \qquad \text{with~$\lambda \in \R$}\:. \]
Then the vector
\[ \phi := S \overline{\phi} \]
is an eigenvector corresponding to the eigenvalue~$-\lambda$ because
\[ \mathscr{C} \phi = \mathscr{C} S \overline{\phi} \overset{\eqref{SUrel}}{=} S\, \big( S \mathscr{C} S \big) \overline{\phi} \overset{\eqref{SCrel}}{=}  -S\, \overline{ \mathscr{C} \phi}
= -\lambda S \overline{\phi} = -\lambda \phi \:. \]

We first give the proof in the case that the matrix~$\mathscr{C}$ is invertible.
We proceed inductively. Let~$\psi_1$ be a normalized eigenvector corresponding to the largest eigenvalue,
\[ \mathscr{C} \psi_1 = \lambda_1\, \psi_1 \qquad \text{with~$\lambda_1 > 0$}\:. \]
Then the vector~$\phi_1 := \mathscr{C} \overline{\psi}_1$
is an eigenvector corresponding to the eigenvalue~$-\lambda_1$.
Now we restrict attention to the orthogonal complement of~$\psi_1$ and~$\phi_1$,
let~$\psi_2$ be an eigenvector corresponding to the largest eigenvalue, set~$\phi_2=\mathscr{C} \overline{\psi}_2$, and so on. After~$N$ we have constructed an orthonormal eigenvector basis
\[ \psi_1,\ldots, \psi_N, \phi_1, \ldots, \phi_N \:. \]
Choosing~$\scrU$ as the matrix whose columns are these basis vectors,
this matrix is unitary and diagonalizes~$\mathscr{C}$,~\eqref{Cdiag}. Moreover, the fact that the last~$N$
columns are obtained from the first~$N$ columns by complex conjugation and the action of~$S$
means that~$\scrU$ has the desired form as in~\eqref{scrUdef}.

It remains to consider the case that~$\mathscr{C}$ has a non-trivial kernel. 
Counting dimensions, it is clear that this kernel is even-dimensional.
Our strategy is to show that this kernel can be removed by a small perturbation of the form
\beq \label{Cperturb}
\mathscr{C}_\lambda = \mathscr{C} + \lambda \begin{pmatrix} D & 0 \\ 0 & -\overline{D} \end{pmatrix}
\eeq
with~$D^*=D$ and~$\lambda>0$. Then we can argue as above, and taking the limit~$\lambda \searrow 0$
gives the result. In order to show that the kernel can indeed be removed with a perturbation of the form~\eqref{Cperturb}, we argue as follows. Assume conversely that~$\mathscr{C}$ has a non-trivial kernel
whose dimension cannot be decreased by first order perturbations of the form~\eqref{Cperturb}.
Then the perturbation operator vanishes on the kernel, i.e.\
\[ \la \psi,\, \begin{pmatrix} D & 0 \\ 0 & -\overline{D} \end{pmatrix} \psi \ra_{\C^{2N}}
\qquad \text{for all~$\psi \in \ker \mathscr{C}$}\:, \]
for any choice of the Hermitian $N \times N$-matrix~$D$. A direct computation shows that
then the vectors~$\psi$ and~$S \overline{\psi}$ must coincide up to a phase.
Since the kernel is even-dimensional, it contains two linearly independent vectors~$\psi_1$ and~$\psi_2$.
After multiplying these vectors by a phase, we can arrange that
\[ \psi_1 = \phi_1 := S \overline{\psi_1} \qquad \text{and} \qquad \psi_2 = \phi_2 := S \overline{\psi_2} \:. \]
Moreover, the two vectors
\[ \psi_1 + e^{i \varphi} \psi_2 \qquad \text{and} \qquad 
S \big( \psi_1 + e^{i \varphi} \psi_2 \big) = \psi_1 + e^{-i \varphi} \psi_1 \]
must be linearly independent for any choice of the
phase angle~$\varphi \in \R$. This implies that~$\psi_1$ and~$\psi_2$ must be zero.
This is a contraction, proving the claim.
\QED

We conclude that, after performing a suitable Bogoliubov transformation, the operator~$Q$
takes the form
\[ Q = \bigg\langle \begin{pmatrix} \tilde{\Psi} \\ \tilde{\Psi}^\dagger \end{pmatrix},
\begin{pmatrix} C & 0 \\ 0 & -C \end{pmatrix} \begin{pmatrix} \tilde{\Psi} \\ \tilde{\Psi}^\dagger \end{pmatrix}
\bigg\rangle \]
with~$\tilde{\Psi}$ again as in~\eqref{scrUdef}. Here~$C$ is a diagonal matrix, which we write as
\[ C = \text{diag}\, (c_1, \ldots, c_N) \qquad \text{with} \qquad c_n \geq 0 \:. \]
The operator~$Q$ is diagonal on the Fock space basis generated from the vacuum by acting with
monomials of the creation operators~$\tilde{\Psi}_n$, with the eigenvalues given as the products
of the corresponding eigenvalues of the matrix~$\tilde{\mathscr{C}}$ in~\eqref{tildeCdef}.
More specifically, on the tensor factor of the $n^\text{th}$ particle, the operator~$Q$ can be represented by the matrix
\beq \label{2matrix}
\begin{pmatrix} e^{-c_n} & 0 \\ 0 & e^{c_n} \end{pmatrix}
= 2 \cosh c_n\:
\begin{pmatrix} 1-d_n & 0 \\ 0 & d_n \end{pmatrix}
\eeq
with
\begin{align*}
d_n &= \frac{e^{c_n}}{2 \cosh c_n} = \frac{e^{c_n}}{e^{c_n} + e^{-c_n}}  \\
1-d_n &= \frac{e^{-c_n}}{e^{c_n} + e^{-c_n}} 
\end{align*}
(here the first and second component refer to occupation numbers one and zero of the $n^\text{th}$ particle,
respectively). In particular, the Fock trace in~\eqref{Wdefgen} reduces to a product,
\[ \tr_\Fock \exp(-Q) = \prod_{n=1}^N 2 \cosh c_n 
= \det \big(2 \cosh C \big) = \sqrt{ \det \big( 2 \cosh \tilde{\mathscr{C}} \big)} \:. \]
With the same method, we can also compute the reduced one-particle density operator~$\tilde{T}$ defined by
\[ \tilde{T}_{jk} = \tr_\Fock \bigg\{ \bigg| \begin{pmatrix} \tilde{\Psi} \\ \tilde{\Psi}^\dagger \end{pmatrix} \bigg\rangle
\bigg\langle \begin{pmatrix} \tilde{\Psi} \\ \tilde{\Psi}^\dagger \end{pmatrix} \bigg|_{jk} \, W \bigg\} \:. \]
Here one picks up the diagonal entries of the matrix in~\eqref{2matrix} to obtain
\[ \tilde{T} = \frac{e^{-\tilde{\mathscr{C}}}}{2 \cosh \tilde{\mathscr{C}}} \:. \]

We finally express our results in terms of the original Fock basis.
Here we use the fact that the Bogoliubov transformation simply corresponds to the
unitary transformation~$\mathscr{C} \rightarrow \tilde{\mathscr{C}} :=
\scrU \mathscr{C} \scrU^{-1}$. Transforming back, we obtain the following result.
\begin{Prp} \label{prpgauss}
Consider the Gaussian state~$W$ as defined by~\eqref{Wdefgen} with~$Q$ as in~\eqref{QCdef}
and a covariance matrix~$\mathscr{C}$ satisfying~\eqref{SCrel}.
Then~$W$ can be written as
\beq \label{Wprp}
W = \frac{1}{\sqrt{ \det \big( 2 \cosh \mathscr{C} \big)} } \: \exp 
\bigg\langle \begin{pmatrix} \Psi \\ \Psi^\dagger \end{pmatrix} , (-\mathscr{C})
\begin{pmatrix} \Psi \\ \Psi^\dagger \end{pmatrix} \bigg\rangle \:.
\eeq
Moreover, its reduced one-particle density operator takes the form
\beq \label{Tprp}
T = \tr_\Fock \bigg\{ \bigg| \begin{pmatrix} \Psi \\ \Psi^\dagger \end{pmatrix} \bigg\rangle
\bigg\langle \begin{pmatrix} \Psi \\ \Psi^\dagger \end{pmatrix} \bigg| \, W \bigg\} = 
\frac{e^{-\mathscr{C}}}{2 \cosh \mathscr{C}} \:.
\eeq
\end{Prp}

\subsection{Expressing Entropies in terms of reduced one-particle density operators}
Using the formulas of Proposition~\ref{prpgauss}, we can now express the von-Neumann
entropy and the relative entropy a Gaussian state in terms of the reduced one-particle density operator.

\begin{Prp} \label{prpvN}
The von-Neumann entropy of the Gaussian state~$W$ in~\eqref{Wprp} can be
expressed in terms of the reduced one-particle density operator~$T$ in~\eqref{Tprp} by
\[ S(W) := -\tr_\Fock\big(W \log W) = -\tr_\H \big( T \log T \big) \:. \]
\end{Prp}
\Proof We first take the logarithm of~\eqref{Wprp},
\beq \label{logW}
\log W = -\frac{1}{2}\: \log \det \big( 2 \cosh \mathscr{C} \big) -
\bigg\langle \begin{pmatrix} \Psi \\ \Psi^\dagger \end{pmatrix} , \mathscr{C}
\begin{pmatrix} \Psi \\ \Psi^\dagger \end{pmatrix} \bigg\rangle \:.
\eeq
Hence
\begin{align}
\tr_\Fock & \big(W \log W) + \frac{1}{2}\: \log \det \big( 2 \cosh \mathscr{C} \big) \notag \\
&= -\tr_\Fock\Big\{
\bigg\langle \begin{pmatrix} \Psi \\ \Psi^\dagger \end{pmatrix} , \mathscr{C}
\begin{pmatrix} \Psi \\ \Psi^\dagger \end{pmatrix} \bigg\rangle\:W \Big\}
= -\tr_\H \big( \mathscr{C} T \big) \:. \label{zwischen1}
\end{align}
Next, taking the logarithm of~\eqref{Tprp}, we obtain
\begin{align}
\log T &= -\mathscr{C} - \log (2 \cosh \mathscr{C}) \\
\mathscr{C} &=  -\log T - \log (2 \cosh \mathscr{C}) \:. \label{ClogT}
\end{align}
Using the last relation in~\eqref{zwischen1}, we conclude that
\begin{align}
S(W) &= -\tr_\Fock\big(W \log W) \notag \\
&= \frac{1}{2}\: \tr_\H \big( \log( 2 \cosh \mathscr{C}) \big) + \tr_\H \big( \mathscr{C} T) \big) \notag \\
&= \frac{1}{2}\: \tr_\H \big( \log( 2 \cosh \mathscr{C}) \big)
- \tr_\H \big( T \log T \big) - \tr_\H \big( T \,\log (2 \cosh \mathscr{C}) \big) \:. \label{zwischen2}
\end{align}
In order to simplify the last trace, we note that, according to~\eqref{tildeCdef}, the eigenvalues of~$\mathscr{C}$
always come in pairs~$\pm c_n$. Since the hyperbolic cosine is even, whereas the corresponding
eigenvalues of~$T$ add up to one (see~\eqref{Tprp}), we obtain
\beq \label{traceformula}
\tr_\H \big( T \,\log (2 \cosh \mathscr{C}) \big) = \frac{1}{2}\: \tr_\H \big(\log (2 \cosh \mathscr{C}) \big) \:.
\eeq
Using this relation in~\eqref{zwischen2} gives the result.
\QED

\begin{Thm} \label{thmgauss} Let~$W$ be the Gaussian state~\eqref{Wprp}
and~$W_0$ another Gaussian state of the same form (denoted everywhere with an additional subscript zero).
Then the relative entropy can be expressed in terms of the reduced one-particle density operators by
\begin{align}
S(W \| W_0) &:= \tr_\Fock \big(W \:(\log W-\log W_{0}) \big) \notag \\
&\:= \tr_\H \Big\{ T \big( \log T - \log T_0 \big) \Big\} - \tr_\H \Big\{ \big( T -T_0 \big) \,\log (2 \cosh \mathscr{C}_0) \Big\} \label{Srel0} \\
&\:= -S(W) + S(W_0) + \tr_\H \Big\{ (T - T_0) \:\mathscr{C}_0 \Big\} \:. \label{Srel1}
\end{align}
If~$\mathscr{C}$ and~$\mathscr{C}_0$ are unitarily equivalent, we also have
\beq \label{Srel2}
S(W \| W_0 )= \tr_\H \Big\{ (T - T_0) \:\mathscr{C}_0 \Big\} \:.
\eeq
\end{Thm}
\Proof First, using~\eqref{logW} and~\eqref{ClogT} with additional subscripts zero, we obtain
\begin{align*}
&\tr_\Fock\big(W \log W_0) + \frac{1}{2}\: \log \det \big( 2 \cosh \mathscr{C}_0 \big) \\
&= 
\tr_\Fock\Big\{
\bigg\langle \begin{pmatrix} \Psi \\ \Psi^\dagger \end{pmatrix} , \mathscr{C}_0
\begin{pmatrix} \Psi \\ \Psi^\dagger \end{pmatrix} \bigg\rangle\:W \Big\}
= \tr_\H \big( \mathscr{C}_0 \, T \big) \\
&= \tr_\H \big( T \log T_0 \big) + \tr_\H \big( T \,\log (2 \cosh \mathscr{C}_0) \big) \:.
\end{align*}
Thus
\begin{align*}
\tr_\Fock\big(W \log W_0)
&=\tr_\H \big( T \log T_0 \big) + \tr_\H \bigg\{ \Big( T -\frac{\1}{2} \Big) \,\log (2 \cosh \mathscr{C}_0) \bigg\} \\
S(W \| W_0) 
&= \tr_\H \Big\{ T \big( \log T - \log T_0 \big) \Big\} - \tr_\H \bigg\{ \Big( T -\frac{\1}{2} \Big) \,\log (2 \cosh \mathscr{C}_0) \bigg\} \:.
\end{align*}
Using~\eqref{traceformula}, this can be rewritten as
\begin{align*}
&S(W \| W_0)
= \tr_\H \Big\{ T \big( \log T - \log T_0 \big) \Big\} - \tr_\H \Big\{ \big( T -T_0 \big) \,\log (2 \cosh \mathscr{C}_0) \Big\} \\
&= -S(W) + S(W_0)
-\tr_\H \Big\{ (T - T_0) \:\log T_0 \Big\} - \tr_\H \Big\{ \big( T -T_0 \big) \,\log (2 \cosh \mathscr{C}_0) \Big\} \\
&= -S(W) + S(W_0)-\tr_\H \Big\{ (T - T_0) \:\big( \log T_0 + \log (2 \cosh \mathscr{C}_0) \big) \Big\} \:.
\end{align*}
The first line gives~\eqref{Srel0}. Moreover, using~\eqref{ClogT} for~$\mathscr{C}_0$ 
in the last line gives~\eqref{Srel1}.

If~$\mathscr{C}$ and~$\mathscr{C}_0$ are unitarily equivalent, the same is
true for~$T$ and~$T_0$. Therefore, the result of Proposition~\ref{prpvN} shows that the
first two summands in~\eqref{Srel1} cancel each other, giving~\eqref{Srel2}.
\QED

\Thanks{{{\em{Acknowledgments:}} We would like to thank the German Science Foundation (DFG) for support. Furthermore, we would like to thank Maximilian Chtchekourov, Stefano Galanda, Robert Jonsson, Simone Murro and Rainer Verch for helpful discussions related to the subject. We also extend our gratitude to the anonymous referee for their valuable feedback and constructive comments, which significantly improved the current form of the manuscript.
\section*{Data Availability Statement}
No datasets were generated or analyzed during the current study.
\section*{Conflict of Interest Statement}
The authors declare that they have no conflicts of interest to disclose.

\bibliographystyle{amsplain}

\begin{thebibliography}{10}

\bibitem{araki1970quasifree}
H.~Araki, \emph{On quasifree states of {${\rm CAR}$} and {B}ogoliubov
  automorphisms}, Publ. Res. Inst. Math. Sci. \textbf{6} (1970/71), 385--442.

\bibitem{Araki}
\bysame, \emph{{Relative Entropy of States of Von Neumann Algebras}}, Publ.
  Res. Inst. Math. Sci. Kyoto \textbf{1976} (1976), 809--833.

\bibitem{bisognano-wichmann1}
J.J. Bisognano and E.H. Wichmann, \emph{On the duality condition for a
  {H}ermitian scalar field}, J. Mathematical Phys. \textbf{16} (1975),
  985--1007.

\bibitem{bisognano-wichmann2}
\bysame, \emph{On the duality condition for quantum fields}, J. Mathematical
  Phys. \textbf{17} (1976), no.~3, 303--321.

\bibitem{bollmann-mueller2}
L.~Bollmann and P.~M\"uller, \emph{Enhanced area law in the {W}idom-{S}obolev
  formula for the free {D}irac operator in arbitrary dimension},
  \href{https://arxiv.org/abs/2405.14356}{arXiv:2405.14356 [math-ph]} (2024).

\bibitem{bollmann-mueller}
\bysame, \emph{The {W}idom-{S}obolev formula for discontinuous matrix-valued
  symbols}, \href{https://arxiv.org/abs/2311.06036}{arXiv:2311.06036
  [math.SP]}, J. Funct. Anal. \textbf{287} (2024), no.~12, Paper No. 110651,
  54.

\bibitem{borchers}
H.~J. Borchers, \emph{{On revolutionizing quantum field theory with Tomita's
  modular theory}}, J. Math. Phys. \textbf{41} (2000), 3604--3673.

\bibitem{cadamuro-froeb}
D.~Cadamuro, M.B. Fr\"{o}b, and C.~Minz, \emph{Modular hamiltonian for fermions
  of small mass}, \href{https://arxiv.org/abs/2312.04629}{arXiv:2312.04629
  [hep-th]} (2023).

\bibitem{casini}
H.~Casini, \emph{Relative entropy and the {B}ekenstein bound},
  \href{https://arxiv.org/abs/0804.2182}{arXiv:0804.2182 [hep-th]}, Classical
  Quantum Gravity \textbf{25} (2008), no.~20, 205021, 12.

\bibitem{casini-huerta}
H.~Casini and M.~Huerta, \emph{Analytic results on the geometric entropy for
  free fields}, \href{https://arxiv.org/abs/0707.1300}{arXiv:0707.1300
  [hep-th]}, J. Stat. Mech. Theory Exp. (2008), no.~1, P01012, 9.

\bibitem{casini-huerta2}
\bysame, \emph{Reduced density matrix and internal dynamics for multicomponent
  regions}, \href{https://arxiv.org/abs/0903.5284}{arXiv:0903.5284 [hep-th]},
  Classical Quantum Gravity \textbf{26} (2009), no.~18, 185005, 15.
  
\bibitem{CiolliLongoRuzzi2020}
F.~Ciolli, R.~Longo, and G.~Ruzzi,
``The information in a wave,''
Commun. Math. Phys. \textbf{379}, 979--1000 (2020),
arXiv:1906.01707 [math-ph].

\bibitem{dangelo}
E.~D'Angelo, \emph{Entropy for spherically symmetric, dynamical black holes
  from the relative entropy between coherent states of a scalar quantum field},
  \href{https://arxiv.org/abs/2105.04303}{arXiv:2105.04303 [gr-qc]}, Classical
  Quantum Gravity \textbf{38} (2021), no.~17, Paper No. 175001, 22.

\bibitem{dangelo-froeb}
E.~D'Angelo, M.B. Fr\"{o}b, S.~Galanda, P.~Meda, A.~Much, and K.~Papadopoulos,
  \emph{Entropy-area law and temperature of de {S}itter horizons from modular
  theory}, \href{https://arxiv.org/abs/2311.13990}{arXiv:2311.13990 [hep-th]},
  PTEP. Prog. Theor. Exp. Phys. (2024), no.~2, Paper No. 021A01, 11.
  \MR{4701668}

\bibitem{derzinski-gerard-quantum}
J.~Derezi\'nski and C.~G\'erard, \emph{Mathematics of {Q}uantization and
  {Q}uantum {F}ields}, Cambridge Monographs on Mathematical Physics, Cambridge
  University Press, Cambridge, 2013.

\bibitem{fermientropy}
F.~Finster, R.~Jonsson, M.~Lottner, A.~Much, and S.~Murro, \emph{Notions of
  fermionic entropies for causal fermion systems},
  \href{https://arxiv.org/abs/2408.01710}{arXiv:2408.01710 [math-ph]} (2024).

\bibitem{intro}
F.~Finster, S.~Kindermann, and J.-H. Treude, \emph{{C}ausal {F}ermion
  {S}ystems: {A}n {I}ntroduction to {F}undamental {S}tructures, {M}ethods and
  {A}pplications}, \href{https://arxiv.org/abs/2411.06450}{arXiv:2411.06450
  [math-ph]}, Cambridge Monographs on Mathematical Physics, Cambridge
  University Press, 2025.

\bibitem{bhentropy}
F.~Finster and M.~Lottner, \emph{The fermionic entanglement entropy of the
  vacuum state of a {S}chwarzschild black hole horizon},
  \href{https://arxiv.org/abs/2302.07212}{arXiv:2302.07212 [math-ph]}, Ann.
  Henri Poincar\'{e} \textbf{26} (2025), 527--595.

\bibitem{diamondentropy}
F.~Finster, M.~Lottner, A.~Much, and S.~Murro, \emph{The fermionic entanglement
  entropy of a causal diamond in two dimensions},
  \href{https://arxiv.org/abs/2407.05292}{arXiv:2407.05292 [math-ph]} (2024).

\bibitem{arealaw}
F.~Finster, M.~Lottner, and A.V. Sobolev, \emph{The fermionic entanglement
  entropy and area law for the relativistic {D}irac vacuum state},
  \href{https://arxiv.org/abs/2310.03493}{arXiv:2310.03493 [math-ph]}, Adv.
  Theor. Math. Phys. \textbf{28} (2024), no.~6, 1933--1985.

\bibitem{rindler}
F.~Finster, S.~Murro, and C.~R\"oken, \emph{The fermionic signature operator
  and quantum states in {R}indler space-time},
  \href{https://arxiv.org/abs/1606.03882}{arXiv:1606.03882 [math-ph]}, J. Math.
  Anal. Appl. \textbf{454} (2017), no.~1, 385--411.

\bibitem{octonions}
F.~Finster, Gresnigt. N.G., J.M. Isidro, A.~Marcian{\`o}, C.F. Paganini, and
  T.P. Singh, \emph{Causal fermion systems and octonions},
  \href{https://arxiv.org/abs/2403.00360}{arXiv:2403.00360 [math-ph]},
  Fortschr. Phys. \textbf{72} (2024), 2400055.

\bibitem{froeb}
M.B. Fr\"{o}b, \emph{Modular {H}amiltonian for de {S}itter diamonds},
  \href{https://arxiv.org/abs/2308.14797}{arXiv:2308.14797 [hep-th]}, J. High
  Energy Phys. (2023), no.~12, Paper No. 74, 57.

\bibitem{froeb-much}
M.B. Fr\"{o}b, A.~Much, and K.~Papadopoulos, \emph{Relative entropy in de
  {S}itter spacetime is a {N}oether charge},
  \href{https://arxiv.org/abs/2310.12185}{arXiv:2310.12185 [gr-qc]}, Phys. Rev.
  D \textbf{108} (2023), no.~10, Paper No. 105004, 19.

\bibitem{galanda2023relative}
S.~Galanda, A.~Much, and R.~Verch, \emph{Relative entropy of fermion excitation
  states on the {CAR} algebra},
  \href{https://arxiv.org/abs/2305.02788}{arXiv:2305.02788 [math-ph]}, Math.
  Phys. Anal. Geom. \textbf{26} (2023), 21.

\bibitem{haagloc}
R.~Haag, \emph{Local quantum physics: Fields, particles, algebras}, Theoretical
  and Mathematical Physics, Springer Berlin Heidelberg, 2012.

\bibitem{helling-leschke-spitzer}
R.~Helling, H.~Leschke, and W.~Spitzer, \emph{A special case of a conjecture by
  {W}idom with implications to fermionic entanglement entropy},
  \href{https://arxiv.org/abs/0906.4946}{arXiv:0906.4946 [math-ph]}, Int. Math.
  Res. Not. IMRN (2011), no.~7, 1451--1482.

\bibitem{hollands-entropy}
S.~Hollands, \emph{Relative entropy for coherent states in chiral {CFT}},
  \href{https://arxiv.org/abs/1903.07508}{arXiv:1903.07508 [hep-th]}, Lett.
  Math. Phys. \textbf{110} (2020), no.~4, 713--733.

\bibitem{hollands-ishibashi}
S.~Hollands and A.~Ishibashi, \emph{News versus information},
  \href{https://arxiv.org/abs/1904.00007}{arXiv:1904.00007 [gr-qc]}, Classical
  Quantum Gravity \textbf{36} (2019), no.~19, 195001, 11.

\bibitem{hollands-sanders}
S.~Hollands and K.~Sanders, \emph{Entanglement {M}easures and their
  {P}roperties in {Q}uantum {F}ield {T}heory},
  \href{https://arxiv.org/abs/1702.04924}{arXiv:1702.04924 [quant-ph]},
  Springer Briefs in Mathematical Physics, vol.~34, Springer, Cham, 2018.

\bibitem{Hollands_2019}
Stefan Hollands, \emph{Relative entropy for coherent states in chiral cft},
  Letters in Mathematical Physics \textbf{110} (2019), no.~4, 713–733.

\bibitem{klich}
I.~Klich, \emph{Lower entropy bounds and particle number fluctuations in a
  {F}ermi sea},
  \href{https://arxiv.org/abs/quant-ph/0406068}{arXiv:quant-ph/0406068}, J.
  Phys. A \textbf{39} (2006), no.~4, L85--L91.

\bibitem{kurpicz}
F.~Kurpicz, N.~Pinamonti, and R.~Verch, \emph{Temperature and entropy-area
  relation of quantum matter near spherically symmetric outer trapping
  horizons}, \href{https://arxiv.org/abs/2102.11547}{arXiv:2102.11547 [gr-qc]},
  Lett. Math. Phys. \textbf{111} (2021), no.~4, Paper No. 110, 44.

\bibitem{La_Piana_2025}
Francesca La~Piana and Gerardo Morsella, \emph{The fermionic massless modular
  hamiltonian}, Communications in Mathematical Physics \textbf{406} (2025),
  no.~4.

\bibitem{leschke-sobolev-spitzer}
H.~Leschke, A.V. Sobolev, and W.~Spitzer, \emph{Trace formulas for
  {W}iener-{H}opf operators with applications to entropies of free fermionic
  equilibrium states}, \href{https://arxiv.org/abs/1605.04429}{arXiv:1605.04429
  [math.SP]}, J. Funct. Anal. \textbf{273} (2017), no.~3, 1049--1094.

\bibitem{LSS_2022}
\bysame, \emph{R\'{e}nyi entropies of the free {F}ermi gas in multi-dimensional
  space at high temperature}, Toeplitz operators and random matrices---in
  memory of {H}arold {W}idom,
  \href{https://arxiv.org/abs/2201.11087}{arXiv:2201.11087 [math-ph]}, Oper.
  Theory Adv. Appl., vol. 289, Birkh\"{a}user/Springer, Cham, 2022,
  pp.~477--508.

\bibitem{longo}
R.~Longo, \emph{Entropy of coherent excitations},
  \href{https://arxiv.org/abs/1901.02366}{arXiv:1901.02366 [math-ph]}, Lett.
  Math. Phys. \textbf{109} (2019), no.~12, 2587--2600.
\bibitem{LongoMorsella2023}
R.~Longo and G.~Morsella,
``The massless modular Hamiltonian,''
Commun. Math. Phys. \textbf{400}, 1181--1201 (2023),
arXiv:2012.00565 [math-ph].

\bibitem{longo-xu}
R.~Longo and F.~Xu, \emph{Relative entropy in {CFT}},
  \href{https://arxiv.org/abs/1712.07283}{arXiv:1712.07283 [math]}, Adv. Math.
  \textbf{337} (2018), 139--170.

\bibitem{ohya-petz}
M.~Ohya and D.~Petz, \emph{Quantum {E}ntropy and its {U}se}, Texts and
  Monographs in Physics, Springer-Verlag, Berlin, 1993.

\bibitem{Uhlmann}
A.~Uhlmann, \emph{Relative entropy and the wigner-yanase-dyson-lieb concavity
  in an interpolation theory}, Communications in Mathematical Physics
  \textbf{54} (1977), no.~1, 21--32.

\bibitem{widom1}
H.~Widom, \emph{On a class of integral operators with discontinuous symbol},
  Toeplitz centennial ({T}el {A}viv, 1981), Operator Theory: Advances and
  Applications, vol.~4, Birkh\"{a}user, Basel-Boston, Mass., 1982,
  pp.~477--500.

\bibitem{witten-entangle}
E.~Witten, \emph{On entanglement properties of quantum field theory},
  \href{https://arxiv.org/abs/1803.04993}{arXiv:1803.04993 [hep-th]}, Rev. Mod.
  Phys. \textbf{90} (2018), 045003.





\end{thebibliography}
\providecommand{\bysame}{\leavevmode\hbox to3em{\hrulefill}\thinspace}
\providecommand{\MR}{\relax\ifhmode\unskip\space\fi MR }
\providecommand{\MRhref}[2]{%
  \href{http://www.ams.org/mathscinet-getitem?mr=#1}{#2}
}
\providecommand{\href}[2]{#2}

\end{document}